\documentclass[]{article}

\usepackage{graphicx} 
\usepackage{amsmath}  
\usepackage{amsfonts} 
\usepackage{import} 
\usepackage[parfill]{parskip}

\usepackage[margin=2.5cm]{geometry}
\usepackage{amsthm} 
\usepackage[affil-it]{authblk} 
\usepackage{amsopn} 
\usepackage{amssymb} 
\usepackage{mathrsfs} 
\usepackage{booktabs} 
\usepackage{natbib} 
\usepackage{tikz} 
\usepackage{algorithm} 
\usepackage{float} 
\usepackage{caption} 
\usepackage{abstract} 
\usepackage{listings} 
\usepackage{color} 
\usepackage{enumerate} 
\usepackage{fancyvrb} 
\usepackage{float}
\usepackage{subcaption}

\usepackage[]{hyperref} 

\newfloat{algorithm}{t}{lop}

\newcommand{\NN}{\mathbb{N}}
\newcommand{\ZZ}{\mathbb{Z}}
\newcommand{\RR}{\mathbb{R}}
\newcommand{\PP}[1]{\mathbb{P}\left( #1 \right)}
\newcommand{\EE}[1]{\mathbb{E} \left[ #1 \right]}

\newcommand{\de}{\mathrm{d}}

\newcommand{\cov}[1]{\, {\rm cov}\left( #1 \right) }
\newcommand{\var}[1]{\, {\rm var}\left( #1 \right) }

\DeclareMathOperator*{\argmin}{argmin} 
\DeclareMathOperator*{\argmax}{argmax} 

\newtheorem{lemma}{Lemma}[section]

\newtheorem{proposition}[lemma]{Proposition}

\theoremstyle{definition}

\theoremstyle{remark}

\theoremstyle{theorem}

\usetikzlibrary{decorations.markings}
\usetikzlibrary{graphs,graphs.standard,quotes}

\newcommand{\periodogram}[2]{I}

\newcommand{\acf}{c}

\newcommand{\varhat}[1]{\, {\rm \hat{var}}\left( #1 \right) }
\newcommand{\covhat}[1]{\, {\rm \hat{cov}}\left( #1 \right) }

\newcommand{\expbit}{\exp\left \{-\frac{r}{s}\left(\frac{\omega}{\omega_p}\right)^{-s}\right \}}
\newcommand{\parbit}{\omega\mid\theta}
\newcommand{\minusparbit}{-\omega\mid\theta}
\newcommand{\dbit}{\delta(\parbit)}
\newcommand{\sigbit}{\sigma(\parbit)}
\newcommand{\fbit}{f_G(\parbit)}

\newcommand{\Sbit}{S_G(\parbit)}
\newcommand{\pard}[1]{\frac{\partial}{\partial #1}}
\newcommand{\negderiv}[1]{\pard{#1}\fbit&=\pard{#1}f_G(\minusparbit)}

\newcommand{\XN}{\bX_{\Delta,N}}
\newcommand{\bX}{\boldsymbol{X}}

\newcommand{\matlab}{MATLAB }
\usepackage{soul}

\begin{document}

\title{Estimating the parameters of ocean wave spectra}
\author[1]{Jake P. Grainger}
\author[2]{Adam M. Sykulski}
\author[2,3]{Philip Jonathan}
\author[4]{Kevin Ewans}
\affil[1]{STOR-i Centre for Doctoral Training, Department of Mathematics and Statistics, Lancaster University, Lancaster, UK}
\affil[2]{Department of Mathematics and Statistics, Lancaster University, Lancaster,UK} 
\affil[3]{Shell Research Ltd., London, UK} 
\affil[4]{MetOcean Research Ltd., New Plymouth, New Zealand}

\date{}

\maketitle

\begin{abstract}
    Wind-generated waves are often treated as stochastic processes.
    There is particular interest in their spectral density functions, which are often expressed in some parametric form.
    Such spectral density functions are used as inputs when modelling structural response or other engineering concerns.
    Therefore, accurate and precise recovery of the parameters of such a form, from observed wave records, is important.
    Current techniques are known to struggle with recovering certain parameters, especially the peak enhancement factor and spectral tail decay.
    We introduce an approach from the statistical literature, known as the de-biased Whittle likelihood, and address some practical concerns regarding its implementation in the context of wind-generated waves. 
    We demonstrate, through numerical simulation, that the de-biased Whittle likelihood outperforms current techniques, such as least squares fitting, both in terms of accuracy and precision of the recovered parameters. We also provide a method for estimating the uncertainty of parameter estimates.
    We perform an example analysis on a data-set recorded off the coast of New Zealand, to illustrate some of the extra practical concerns that arise when estimating the parameters of spectra from observed data.
\end{abstract}

\section{Introduction}
Due to the random nature of wind-generated gravity waves, it is common to treat them as stochastic processes. 
There is particular interest in the spectral density function of such wave processes.
For this reason, it is important that we are able to construct good spectral density estimators.
Using such an estimator, estimates of the spectral density function can be obtained from observed wave records.
Broadly speaking, there are two approaches for obtaining such an estimator: {non-parametric} and {parametric}. 
The most basic {non-parametric} spectral density estimator is the periodogram --- the Fourier transform of the sample autocovariance. 
However, the periodogram is a noisy estimator. Therefore many less noisy estimators, such as that of \cite{Bartlett1948}, have been developed. 
The second approach is to use a {parametric} spectral density estimator. Here we assume that the spectral density function follows a parametric form, meaning that the inference task becomes estimation of the parameters of this form.
In general, parametric estimators are often preferable because they result in smoother estimates and more concise representations of the spectral density function---and the parameters themselves provide physical interpretation of the nature of the wave process.
\par
Many such parametric forms have been developed in the oceanography literature.
\cite{Phillips1958} gave theoretical arguments for the tail behaviour of the spectral density function for wind-generated wave processes.
Based on this, \cite{Pierson1964} established a parametric form that characterised the spectral density function of a fully developed sea, describing both the spectral tail and peak behaviour.
This was later extended by \cite{Hasselmann1973}, so that the parametric form could encompass a wider variety of spectral density functions, including those associated with fetch limited wave processes.
This widely used parametric form is usually known as the JONSWAP spectral form. It should be noted that we use JONSWAP to refer to the original formulation given by \cite{Hasselmann1973}, with a tail decay of $\omega^{-5}$ (where $\omega$ denotes angular frequency).
\par
Despite general acceptance of the JONSWAP spectral form amongst practitioners, there is debate concerning the values of the tail decay index and peak enhancement factor.
Arguments for an $\omega^{-5}$ tail decay, made by \cite{Phillips1958}, were called into question by \cite{Toba1973} and later by \cite{Phillips1985}, who argued that an $\omega^{-4}$ tail had a stronger theoretical basis. Experimental work such as \cite{Hasselmann1973} and \cite{Battjes1987a} found evidence for both $\omega^{-4}$ and $\omega^{-5}$ tail decays, while \cite{Hwang2017} could not find evidence for either, further suggesting that the tail decay index should be treated as a free parameter.
In addition, there is a large literature speculating on other tail behaviours, such as the occurrence of a transition frequency from $\omega^{-4}$ to $\omega^{-5}$ \citep[for example]{Forristall1981,Ewans1986,Babanin2010}.
It is also common to fix the peak enhancement factor to 3.3; however, there is little evidence for using precisely this value.
In this work, we adopt a more general version of the JONSWAP spectral form, which treats both the tail decay index and peak enhancement factor as free parameters
(though our methods also apply to the special cases mentioned, in terms of estimating the remaining parameters of interest).\footnote{To avoid potential confusion, it should be noted that we are interested in estimating the parameters of assumed parametric forms for the spectral density function (such as the generalised JONSWAP) and not spectral parameters of a sea state such as significant wave height ($H_s$).}
Many authors \citep[for example]{Rodriguez1999,Ewans2018} have found that both the tail decay index and peak enhancement factor are hard to estimate accurately, using current techniques.
However, both of these parameters are important for determining the properties of a given sea state.
Our contention is that current techniques are not sufficiently accurate or precise to allow strong statements to be made concerning the true values of the tail decay index or peak enhancement factor, from typical data sets.
Indeed, in Section \ref{sec:sim} we demonstrate with simulated half hour records that estimates for the tail decay (using current estimation techniques) range from $\omega^{-3}$ to $\omega^{-6}$, when the true tail decay is known to be $\omega^{-4}$. 
Because there is too much variability in the estimates, it is impossible to determine from data of such lengths if the true tail decay is $\omega^{-4}$ or $\omega^{-5}$.
In this work we present an alternative technique that is capable of obtaining these parameters more accurately and precisely, and show in simulated data that this technique can distinguish $\omega^{-4}$ and $\omega^{-5}$ tail decays, even from short records.
\par
The standard approach for estimating parameters of a stochastic model from data is by using maximum likelihood inference.
When an analytical form for the likelihood function is known, such parameters can be optimally estimated using maximum likelihood \citep{Pawitan2001}.
However, in the case of wind-generated wave processes, the exact probability distribution is unknown.
Though it is possible to make the simplifying assumption that the wave process is Gaussian, for many sea states this assumption will not be reasonable.
For this reason, it has become common for oceanographers to use a non-parametric estimator of the spectral density function, and obtain parameters by fitting a parametric form in the least squares sense. However, such least squares estimators will in general be sub-optimal when compared to full maximum likelihood \citep{Constable1988}.
\par
We therefore turn to frequency domain likelihoods, which are widely used in both time series analysis and spatial statistics \citep[for example]{Nordman2006,Fuentes2007}. 
The canonical approach is to use an approximation to maximum likelihood known as the Whittle likelihood \citep{Whittle1953}.
The Whittle likelihood can be computed quickly using Fast Fourier Transforms and does not require Gaussianity \citep{Dzhaparidze1983}. However, the Whittle likelihood has been shown to produce biased estimates for small sample sizes \citep{Dahlhaus1988,Velasco2000}.
\cite{Sykulski2016a} developed a de-biased version of the Whittle likelihood that corrects for this bias, without sacrificing the computational speed or making extra distributional assumptions.
In Section \ref{sec:fitting:comparison}, we will provide some intuition as to why we would expect the de-biased Whittle likelihood to perform better than least squares, both in terms of accuracy (bias) and precision (variance). Then in Section \ref{sec:sim}, we evidence this claim using numerical simulations.
\par
The contributions of this paper are as follows.
Firstly, we introduce the de-biased Whittle likelihood estimator for use on wind-generated wave processes (Section~\ref{sec:fitting}). Secondly, we detail practical concerns regarding the implementation of the estimation procedure for wind-generated ocean wave processes (Section~\ref{sec:practical}), with accompanying \matlab code provided on GitHub \citep{Grainger2021a}. This includes an important generalisation of the \cite{Sykulski2016a} procedure to allow parameters to be fitted directly to the proposed spectral form {\em without} having to posit an analytical form for the time-domain theoretical autocovariance sequence---as required in \cite{Sykulski2016a}, but unavailable for ocean wave spectral forms. Thirdly, we present a novel reformulation of the variance of the de-biased Whittle likelihood estimator which can be used to quantify the uncertainty of parameter estimates (Section~\ref{sec:confidence}). Finally, we perform a detailed simulation and field data study comparing the performance of different parametric spectral density estimators for wind-generated wave processes (Section~\ref{sec:sim} and Section~\ref{sec:application}).

\section{Background}\label{sec:background}
So far we have used the word ``wave'' loosely to describe the shifting nature of the sea surface. In truth, we are actually interested in modelling the displacement of the sea surface from the resting surface. Of course, in reality this is a 3-dimensional phenomena, but in this paper we shall consider the vertical displacement of the surface over time at a specific location in space.
We can think of the displacement at a given time as being a \textit{random variable} with some distribution. Therefore we can describe the displacement over time by a \textit{stochastic process}, an indexed family of random variables, which we shall denote $\bX=\{X_t\}_{t\in \RR}$. Note that this is a family of random variables indexed over continuous time, as the actual physical process is constantly changing. However, since we cannot actually record data continuously in time, we must instead settle for recording the process at discrete points in time. We assume that the data are being sampled regularly and denote the sampling interval $\Delta$ and the process that arises from sampling $\bX$ every $\Delta$ seconds we shall call $\bX_\Delta=\{X_{t\Delta}\}_{t\in\ZZ}$.
\par
For the duration of a given record, observations of the sea surface are usually assumed to be from an underlying process $X$ that is \textit{second-order stationary}. This means $\bX$ satisfies all of the following conditions:
\begin{enumerate}
    \item $\EE{X_t}=\EE{X_0}$,
    \item $\EE{|X_t|^2}<\infty$,
    \item $\EE{X_tX_s}=\EE{X_{t-s}X_0}$,
\end{enumerate}
for all $t,s\in \RR$. 
However, the sea surface is not actually stationary: it evolves over time. One way to circumvent this is to notice that whilst the conditions at sea do evolve over time, they do so relatively slowly if we sample frequently. Therefore, we treat the sea surface as being stationary over short time intervals, sometimes known as sea states. This is essentially the same approach as locally stationary modelling in time series analysis \citep{Dahlhaus2012}. Ideally we would make this sea state as short as possible. However, we must balance this with another fundamental statement: the more observations we have, the more confident we can be in our inferences.
The question of the correct time interval to use will not be covered here; though, it is useful to keep in mind that improving the precision of parameter estimates will mean that we could use shorter sea states in our analysis. This would allow us to track the evolution of certain meteorological processes, such as tropical cyclones, at a higher precision and resolution.
To summarise, the underlying wind-generated wave process is not second-order stationary; however, for short enough time windows, this is a widely used working assumption that allows some inference to be made about the process in question.
\par
The analysis of second-order stationary stochastic processes usually involves two important characteristics: the \textit{autocovariance} and the \textit{spectral density function}. The autocovariance of a process at a given lag $\tau$, is just the covariance of a process with itself $\tau$ time-steps later. More formally, the autocovariance is $\acf(\tau) = \EE{X_{\tau}X_{0}}-\EE{X_\tau}\EE{X_0}$. For our purposes, we assume that $\EE{X_t}=0$ for all $t\in\RR$; noting that if this is not the case, then by first removing the mean of the data we can obtain a process with the desired property. Therefore, the autocovariance simplifies to $\acf(\tau)=\EE{X_{\tau}X_0}$.
The spectral density function is a frequency domain analogue of the autocovariance, which for the stochastic processes $\bX$ and $\bX_\Delta$ we shall denote $f(\omega)$ and $f_\Delta(\omega)$ respectively. A formal construction of the spectral density function can be found in, for example, \citet[page 118]{Brockwell2006} (for the discrete time process), or \citet[page 522]{Doob1953} (for the continuous time process). However, for our purposes it suffices to note the following relations. Firstly for the discrete time process,
\begin{align}
    f_\Delta(\omega)&=\frac{\Delta}{2\pi}\sum_{\tau=-\infty}^\infty \acf(\tau\Delta)\exp\{-i\omega \tau\Delta\},\label{eq:sdfvsacf}
\end{align}	
for $\omega\in[-\pi/\Delta,\pi/\Delta]$, where $\pi/\Delta$ is the \textit{Nyquist frequency} and is the highest observable frequency of the sampled process. Secondly for the continuous time process,
\begin{align}
    f(\omega)&=\frac{1}{2\pi}\int_{-\infty}^\infty \acf(\tau)\exp\{-i\omega \tau\}\de\tau,\label{eq:sdfvsacf:continuous}
\end{align}
for almost every $\omega\in\RR$ (i.e. equal except on a set of measure zero).\footnote{Note that we are working with angular frequency here, and for all examples in this paper this is measured in units of $\text{rad s}^{-1}$.} Similarly, the inverse relations are
\begin{align}
    \acf(\tau\Delta)&=\int_{-\pi/\Delta}^{\pi/\Delta} f_\Delta(\omega) \exp\{i\omega\tau\Delta\}\de\omega,\label{eq:acf:discrete}
\end{align}
for $\tau\in\ZZ$ and
\begin{align}
    \acf(\tau)&=\int_{-\infty}^\infty f(\omega)\exp\{i\omega\tau\}\de\omega,
    \label{eq:acf:continuous}
\end{align}
for $\tau\in\RR$.
The spectral density of the discrete time process, $f_\Delta(\omega)$, can be thought of as an aliased version of the continuous time spectral density function $f(\omega)$. More formally, we have the following relation:
\begin{align}
    f_\Delta(\omega)&=\sum_{k=-\infty}^\infty f\left(\omega+\frac{2\pi k}{\Delta}\right),
    \label{eq:sdf:aliasing}
\end{align}
for $\omega\in[-\pi/\Delta,\pi/\Delta]$  \citep[chapter 4]{Percival1993}. In Section  \ref{sec:sim}, we demonstrate that aliasing can cause bias in parameter estimation, which is why it is important to define both $f(\omega)$ and $f_\Delta(\omega)$ and understand their relationship.
\subsection{Non-parametric spectral density estimators} \label{sec:background:non}
    Though our purpose is the analysis of parametric spectral density estimators, it is also pertinent to define some of the non-parametric spectral density estimators that are used throughout this paper.
    There are two important properties that should be considered when choosing an estimator. The first of these is \textit{bias}, which is the expectation of the estimator minus the true value. Ideally we would want to choose an estimator that is unbiased, i.e. has a bias of zero. This is often not possible, but the weaker condition of asymptotically unbiased is often achievable. An estimator is said to be asymptotically unbiased if, as the number of observations increases, the bias tends to zero.
    The second important property is \textit{consistency}. For an estimator to be consistent it must converge in probability to the true parameter as the number of observations tends to infinity. 
    More formally, denote the true parameter by $\theta_0$ and an estimator from a series of $N$ observations by $\hat\theta_N$. Then $\hat\theta_N$ is a consistent estimator if, for all $\epsilon>0$, $\PP{|\hat\theta_N-\theta_0|>\epsilon}\rightarrow0$ as $N\rightarrow\infty$.
    \par
    The most basic non-parametric estimator for the spectral density function of a discrete time process is the \textit{periodogram}. Let $\XN=\{X_0,X_\Delta,\ldots,X_{\Delta(N-1)}\}$ be a series of $N$ consecutive random variables from $X_\Delta$, then the periodogram is defined as
    \begin{align*}
        \periodogram{\Delta}{N}(\omega)&=\frac{\Delta}{2\pi N}\left|\sum_{t=0}^{N-1}X_{\Delta t} \exp\{-it\Delta\omega\}\right|^2,
    \end{align*}
    for $\omega\in\RR$. In practice, the periodogram is typically only evaluated at the Fourier frequencies
    $\omega=2\pi j/\Delta N$ using the FFT procedure, where $j=-\lceil N/2\rceil+1,\ldots,\lfloor N/2\rfloor$.
    For convenience, we shall write $\Omega_{N,\Delta}$ for the set of these frequencies.
    It should also be noted that the periodogram is an estimator for the spectral density of the discrete time process $f_\Delta(\omega)$, not the spectral density of the continuous time process $f(\omega)$.
    The periodogram can be shown to be an asymptotically unbiased estimator for $f_\Delta(\omega)$, but the periodogram is not consistent. 
    \par
    For this reason, modified versions of the periodogram, which are consistent, are usually used as an alternative to the periodogram. 
    One such modified periodogram, suggested by \cite{Welch1967}, splits the series into smaller segments, applies a window, calculates the periodogram of each segment and then averages these modified periodograms at each frequency. In practice, Welch's method results in an estimate that is less noisy than a standard periodogram, but has lost resolution in frequency and may be more biased.
    A subset of such methods is known as \textit{Bartlett's method} \citep{Bartlett1948}.
    This approach uses non-overlapping segments with no window function. In other words, Bartlett's estimator is
    \begin{align}
        I_B(\omega) &= 
        \frac{\Delta}{2\pi PL}\sum_{p=1}^P\left|\sum_{t=(p-1)L+1}^{pL} X_{\Delta t} \exp\{-it\Delta\omega\}\right|^2,
    \end{align}
    where $P$ is the number of segments and $L$ is the number of observations in each segment (with $PL\leq N$).
\subsection{Models for the sea surface}\label{sec:ocean_model}
When describing the sea surface, models are often expressed in terms of the spectral density function.
Many different spectral density functions have been developed for ocean waves, perhaps most notably the JONSWAP spectra, developed by \cite{Hasselmann1973}. We shall consider a more general model, which encompasses many of the other waves models that have been developed. Following \cite{Mackay2016}, we use the following parametrisation:
\begin{align}
	\Sbit&=\alpha\omega^{-r}\exp\left \{-\frac{r}{s}\left(\frac{\omega}{\omega_p}\right)^{-s}\right \}\gamma^{\dbit},
	\label{eq:JONSWAP}
\end{align}
where
\begin{align*}
	\dbit&=\exp\left\{-\frac{1}{2\sigbit^2}\left (\frac{\omega}{\omega_p}-1\right )^2\right\},
\end{align*}
and
\begin{align*}
\sigbit&=\begin{cases}
	\sigma_1\quad\text{for }\omega\leq\omega_p\,,\\
	 \sigma_2\quad\text{for }\omega>\omega_p\,,
\end{cases}	
\end{align*}
for $\omega >0$; where $\alpha,\omega_p,s>0$, $\gamma\geq1$, $r>1$\footnote{\cite{Mackay2016} gives the condition that $r>0$. However, for the spectral density to be integrable (such that the stochastic process has finite variance), we require that $r>1$.} and $\theta$ denotes the vector of parameters. Typically, and for the remainder of this paper, $\sigma_1,\sigma_2$ and $s$ are set to $0.07,0.09$, and $4$ respectively \citep{Mackay2016}.
In this case, the parameter vector is $\theta=(\alpha,\omega_p,\gamma,r)$.
Also let $\Theta$ denote the \textit{parameter space} --- the set of possible values that $\theta$ can take. Then for this general model, the parameter space is  $\Theta=(0,\infty)\times(0,\infty)\times[1,\infty)\times(1,\infty)\subseteq\RR^4$.
Note that~\eqref{eq:JONSWAP} is a one sided spectral density, and is not defined at $\omega=0$. We shall work with the two sided version as this fits in with the way we have defined the spectral density function, the way techniques are described in the statistical literature, and the way Fast Fourier Transforms are implemented on a computer. Therefore, we define $f_G:\RR\times\Theta\rightarrow[0,\infty),$
\begin{align}
    \fbit&=
        \begin{cases}
	        \Sbit/2 & \text{for }\omega>0,\\
	        0 & \text{for }\omega=0,\\
	        S_G(\minusparbit)/2  & \text{for }\omega<0.
	    \end{cases}\label{eq:genJONSWAP}
\end{align}
We shall refer to the function defined by \eqref{eq:genJONSWAP} as the generalised JONSWAP spectral form. In this formulation, $\alpha$ is measured in units of $\text{m}^2 \text{ s}^{1-r} \text{ rad}^{r-1}$, $\omega_p$ in $\text{rad s}^{-1}$ and $\gamma$ and $r$ are dimensionless. For convenience, we omit the units in future references.
\section{Fitting parametric spectral density functions}\label{sec:fitting}
The process of fitting a parametric spectral density function to observations can be thought of as estimating the parameters of a statistical model, which we denote $\theta$. 
The techniques discussed in this section are applicable to a broad class of spectral density functions. As such, we consider the general case and shall write $f(\omega\mid\theta)$ for the spectral density function of the continuous time process, given some choice of parameters $\theta$. We shall also write $f_\Delta(\omega\mid\theta)$ and $\acf(\tau\mid\theta)$ for the spectral density function of the discrete time process and the autocovariance function respectively. For convenience, we shall sometimes refer to the spectral density function of the continuous time process as the spectral density function, and the spectral density function of the discrete time process as the aliased spectral density function.
We also write $\Sigma_\theta$ for the covariance matrix of the multivariate random variable corresponding to $N$ consecutive random variables from $X_\Delta$. We now describe each of the fitting methods discussed in this paper.
    \subsection{Least squares}
    Current approaches to estimating parameters of spectral density functions used in the ocean waves literature, such as the approaches described by \cite{Ewans2018}, usually involves two key steps.
    Firstly, a non-parametric estimator of the spectral density function is constructed.
    Secondly, a curve fitting algorithm is used so that the corresponding parametric form is a good fit for the observed data. 
    Typically this involves minimising the square distance between the parametric form and non-parametric spectral density estimator. As such, we shall refer to such approaches as least squares fitting techniques.
    \par
    For the purpose of this section, we let $\bar I(\omega)$ denote a general non-parametric spectral density estimator (this could be the periodogram, $I(\omega)$, Bartlett estimator, $I_B(\omega)$, or some other non-parametric spectral density estimator).
    The second part of this fitting routine involves fitting the parametric form to the obtained non-parametric estimator. Typically, this is done by minimising the Euclidean distance between the non-parametric estimator and the parametric spectral density function. We therefore must minimise the objective function given by
    \begin{align}
        \ell_{LS}\left(\theta\mid\XN\right)&=\sum_{\omega\in\Omega} \left(f(\omega\mid\theta)-\bar I(\omega)\right)^2,\label{eq:LS}
    \end{align}
    where $\Omega\subseteq\Omega_{N,\Delta}$ (the choice of $\Omega$ is discussed in Section  \ref{sec:freq_selection}).
    In other words, the least squares estimator for $\theta$ is defined as $\hat \theta_{LS}=\argmin_{\theta\in\Theta} \ell_{LS}\left(\theta\mid\XN\right)$.
    This approach could be adapted to account for aliasing by replacing $f(\omega\mid\theta)$ with $f_\Delta(\omega\mid\theta)$, the aliased spectral density function; however, such a modified approach is not currently used in the ocean waves literature and therefore we shall use the form given by \eqref{eq:LS} in our simulation study.
    \par
    
    Part of the reason that least squares performs poorly is that the variance of a spectral estimate will be different at different frequencies.
    This means that low density areas of the spectral density function (such as the high frequency tail) tend to be under-weighted.
    For this reason, log transforms are often used in least squares objective functions, especially in the statistics literature \citep{bloomfield1973}.  
    Because standard least squares is widely used in the ocean waves literature, we present a comparison of standard least square in this paper.
    However, in simulations not shown in this paper, log least squares still does not perform as well as the de-biased Whittle likelihood. 
    Whilst log least squares does provide better estimates of the spectral tail decay index than standard least squares, some of the other parameter estimates have increased bias when compared to standard least squares.
    Plots of these log least squares simulations are available on GitHub \citep{Grainger2021a}.
    
    \subsection{Maximum likelihood}
    Maximum likelihood inference treats the sea surface data as observations of a random variable with a given distribution. The parameters for this distribution are chosen by maximising the probability of observing the data given that the underlying distribution has certain parameters. 
	For the moment, let the sea surface observations be multivariate Gaussian with expectation zero and an unknown covariance matrix $\Sigma_\theta$.
	The log-likelihood function for observations of such a process is
	\begin{align}
	    \ell_{ML}\left(\theta\mid\XN\right)&= \frac{1}{2}\left(-N\log(2\pi)-\log|\Sigma_\theta|-\XN^{T}\Sigma_\theta^{-1}\XN\right) ,\label{eq:ML}
	\end{align}
	where $\XN^T$ denotes the transpose of $\XN$.
	The maximum likelihood estimator is then obtained by maximising the log-likelihood function. More formally, the maximum likelihood estimator of $\theta$ is
	$
        \hat \theta_{ML}=\argmax_{\theta\in\Theta} \ell_{ML}\left(\theta\mid\XN\right)
    $.
	Provided that the underlying random variable is actually multivariate Gaussian, this technique will provide asymptotically optimal estimates of $\theta$, in the sense that they converge at an optimal rate (see \cite{Pawitan2001}, chapter 8.5, for more details).
	\par
	This approach can be computationally expensive because evaluating the objective function given by \eqref{eq:ML} requires the inversion of an $N\times N$ matrix. 
	Also, if we want to model a distribution that is not Gaussian, then a different log-likelihood function must be used. 
	This may take significantly longer to compute, or may not even be tractable.
	As previously discussed, wave processes will not typically be precisely Gaussian. However, in Section~\ref{sec:sim} we shall compare fitting techniques on simulated Gaussian processes in the first instance. In this case, full maximum likelihood provides a useful benchmark to compare the performance of other estimators to the optimal choice of estimator.

    \subsection{Spectral-likelihood}
    To avoid some of the problems associated with maximum likelihood estimation we can use approximations to the likelihood, known as pseudo- or quasi-likelihoods, to gain some of the accuracy and precision of maximum likelihood, while keeping computational costs (and distributional assumptions) low.
    One such quasi-likelihood\footnote{These likelihoods are usually referred to as quasi-likelihoods. However, we also use the term spectral-likelihood as it integrates nicely with current terminology used in the literature, as well as giving an intuitive sense of what a spectral-likelihood does.} is known as the Whittle likelihood \citep{Whittle1953}. The Whittle likelihood has been used in a wide range of applications due to its computational efficiency and fairly free distributional assumptions (in particular, we no longer need to assume that the underlying process is Gaussian).
    In its discretised form, the Whittle likelihood is
	\begin{align}
	    \ell_{W}\left(\theta\mid\XN\right)&=-\sum_{\omega\in\Omega} \left\{ \log\left(f(\omega\mid\theta)\right)+\frac{\periodogram{\Delta}{N}(\omega)}{f(\omega\mid\theta)}\right\},\label{eq:WL}
	\end{align}
	where $I(\omega)$ denotes the \textit{periodogram} ordinate at angular frequency $\omega$. The corresponding estimator is again obtained by maximising this spectral-likelihood, which we shall denote by
	$
        \hat \theta_{W}=\argmax_{\theta\in\Theta} \ell_{W}\left(\theta\mid\XN\right)
    $.
	This estimator also does not account for aliasing. However, by replacing $f(\omega\mid\theta)$ with $f_\Delta(\omega\mid\theta)$ in \eqref{eq:WL}, we obtain an estimator that does account for aliasing. 
	We shall refer to this as the aliased Whittle likelihood, though it should be noted that some authors refer to this as simply the Whittle likelihood.
	\par
	Though this aliased approach accounts for some of the bias in the Whittle likelihood, other forms of bias introduced through phenomena such as blurring are still present \cite[chapter 6]{Percival1993}. \cite{Sykulski2016a} introduced the de-biased Whittle likelihood to deal with both aliasing and blurring simultaneously. 
	The de-biased Whittle likelihood is
	\begin{align}
	    \ell_{DW}\left(\theta\mid\XN\right)&=-\sum_{\omega\in\Omega} \left\{ \log\left(\bar f_N(\omega\mid\theta)\right)+\frac{\periodogram{\Delta}{N}(\omega)}{\bar f_N(\omega\mid\theta)}\right\},\label{eq:DW}
	\end{align}
	where $\bar f_N(\omega\mid\theta)=\EE{\periodogram{\Delta}{N}(\omega)\mid\theta}$ is the \textit{expected periodogram}.
    As noted by \cite{Sykulski2016a}, the expected periodogram can be calculated in $O(N\log N)$ time by using the relation:
	\begin{align}
	    \EE{\periodogram{\Delta}{N}(\omega)\mid\theta}&=\frac{1}{2\pi}\text{Re}\Bigg(2\Delta\sum_{\tau=0}^{N-1}\left(1-\frac{\tau}{N}\right)\acf(\tau\mid\theta)\exp\{-i\omega\tau\Delta\}-\Delta\acf(0\mid\theta)\Bigg).\label{eq:EI}
	\end{align}
	The resulting estimator can then be expressed as
	$
        \hat \theta_{DW}=\argmax_{\theta\in\Theta} \ell_{DW}\left(\theta\mid\XN\right).
    $
    Despite being constructed from the periodogram, an inconsistent estimator of the spectral density function, the de-biased Whittle likelihood is a consistent estimator of the parameters for the parametric model. The de-biased Whittle likelihood is able to address the deficiencies in the periodogram without introducing bias, by accounting for the finite sample properties of the periodogram.
    \cite{Sykulski2016a} also show that, under certain conditions, the de-biased Whittle estimator converges optimally.
    These condition are discussed further in Appendix~\ref{append:assumption}.
    \subsection{Comparison}\label{sec:fitting:comparison}
In Section~\ref{sec:sim}, we perform a simulation study to compare each of the estimators that we have discussed.

However, we can also try to build some intuition as to why certain approaches are likely to be more effective than others.
To achieve this we shall consider the conditions in which each technique would be equivalent to full maximum likelihood for a finite sample, then evaluate how likely it is that said assumptions are satisfied.
Note that this is not (and nor is it intended to be) a formal proof; results related to the convergence of de-biased Whittle estimators and their proofs can be found in \cite{Sykulski2016a}.
Rather, this is a sketch of what is going on under the hood that causes the de-biased Whittle likelihood to outperform least squares based techniques.

\par

Maximum likelihood inference works by making probabilistic statements about the distribution of data and then using this to work out the parameter choice that would have been most likely to have given rise to the data in question.
The part of this process of interest to us here is making such distributional statements.
For this comparison, we shall think of the non-parametric spectral density estimates as ``the data'', and shall consider what their distribution would need to be for least squares or the Whittle likelihood to be the optimum likelihood function for this data.\footnote{This differs from full maximum likelihood on the time series as we have lost the phase information in calculating a spectral density estimate.}

\par
Then, for the least squares approach to yield the same parameter estimates as the optimum likelihood function, we would need the non-parametric spectral density estimator used in the fitting routine, $\bar I(\omega)$, to satisfy the following four assumptions.
Firstly, at each frequency, the non-parametric estimator must be Gaussian. In general, this is not true for non-parametric spectral density estimators, though it is true asymptotically for some of them (e.g. Bartlett's method).
Secondly, the expectation of the non-parametric estimator must be equal to the spectral density function at a given frequency.
This is not actually true for non-parametric spectral density estimators, as these are constructed to estimate the aliased spectral density function, not the spectral density function of the continuous time process.
Though this aliasing could be accounted for by modifying the spectral form used in the fitting routine, such modification is not standard practice and many non-parametric spectral density estimators are still biased.
Thirdly, the variance of the non-parametric spectral density estimator must be the same for each frequency.
This is not the case for non-parametric spectral density estimators in general, as the variance at a given frequency depends on the spectral density function at that frequency \citep{Brockwell2006}.
Though weighted least squares approaches, such as the approach proposed by \cite{Chiu1988}, do begin to address the problem of assumption three, they are not widely used and still make the first and second assumptions.
Fourthly, the non-parametric estimators at any two different frequencies must be uncorrelated.
This assumption is discussed further in Section~\ref{sec:fitting:differencing}.
\par
For the Whittle likelihood to yield the same parameter estimates as the optimum likelihood function, we would need the following three assumptions on the periodogram to hold.
Firstly, we would require the periodogram to be exponentially distributed at each Fourier frequency.\footnote{A slightly different assumption is made about the zero and Nyquist frequency, though in practice they are often omitted.}
Secondly, we would require that the expectation of the periodogram is equal to the spectral density function at a given frequency (and consequently that the variance is the square of the spectral density function).
Thirdly, the periodogram at any two different frequencies must be uncorrelated.
At fixed frequencies, the first assumption is true asymptotically for linear processes \citep{Brockwell2006} and for some classes of non-linear processes \citep{shao2007}. At first glance, this may seem to be similar to the asymptotic normality of Bartlett modified periodograms that are often used in least squares.
However, it should be noted that in the case of the periodogram, this asymptotic result is in terms of the number of observations; whereas for Bartlett modified periodograms, this result is in terms of the number of segments that are used, which is much smaller.
When it comes to the second assumption, the periodogram is an asymptotically unbiased estimator of the aliased spectral density function. For this reason, the aliased version of the Whittle likelihood should be used over the standard version. Again it may seem that this is also true for Bartlett modified periodograms, as Bartlett's method averages periodograms and each periodogram is an asymptotically unbiased estimator of the aliased spectral density function.
Therefore if we were to adjust for this aliasing, least squares would be justified.
However, each of these component periodograms are calculated from small segments of the full record, so it is difficult to invoke asymptotic results.
Indeed, this creates somewhat of a catch-22 for Bartlett least squares: to get asymptotic normality we must average many periodograms; but this results in using shorter segments for each periodogram, introducing bias (and vice versa). 
The de-biased Whittle likelihood \citep{Sykulski2016a} bypasses the second assumption (made by the Whittle likelihood) altogether, as it uses the theoretical expectation of the periodogram in place of the spectral density function.
This means that even for small sample sizes the de-biased Whittle likelihood produces estimates with very small to no bias.
The final assumption, the assumption of independence between frequencies, is required by both least squares and spectral-likelihoods; however, the Whittle likelihood is also in a strong position when it comes to satisfying this assumption. This is because asymptotically the periodogram is uncorrelated at different frequencies, and we are using the longest periodogram possible, given the length of the data.
Of course, least squares techniques could be used on the raw periodogram, meaning that the second and last assumptions are just as likely to be satisfied as when using spectral-likelihoods, but in this case, the asymptotic normality required for least squares will not be satisfied (nor in general will the constant variance assumption).
\par
When it comes to the final assumption for both least squares and spectral-likelihood techniques, there are some practical concerns that should be considered.
In particular, when the aliased spectral density has high dynamic range, the frequencies are often correlated.
As we shall shortly show, in the case of wind-generated waves, this issue does not present itself for 1Hz data. Although, for higher sampling frequencies, such as 4Hz data, the periodogram is often highly correlated.
To solve this problem we can turn to differencing, a technique that is well established for reducing correlations in the periodogram \citep{Velasco2000}.
\subsection{Differencing}\label{sec:fitting:differencing}
If the periodogram is highly correlated across frequencies, spectral-likelihoods will perform poorly when compared to full maximum likelihood \cite{Velasco2000}.
Differencing can sometimes provide a convenient mechanism for removing such correlations.
Define the differenced process as $Y_t=X_{t+\Delta}-X_t$.
We briefly switch notation and let $\acf_X(\tau)$ and $f_X(\omega)$ denote the autocovariance and spectral density function of $X$ at $\tau$ and $\omega$ respectively, and likewise $\acf_Y(\tau)$ and $f_Y(\omega)$ for the differenced process $Y$.
First notice that
\begin{align}
    \acf_Y(\tau) 
    &= \EE{Y_0Y_\tau}\\
    &= \EE{X_\Delta X_{\tau+\Delta}-X_0X_{\tau+\Delta}-X_\Delta X_{\tau}+X_0X_\tau}\\
    &= 2\acf_X(\tau)-\acf(\tau+\Delta)-\acf(\tau-\Delta),
\end{align}
by stationarity.
Then from \eqref{eq:sdfvsacf:continuous} we can see that
\begin{align}
    f_Y(\omega)
    &=\int_{-\infty}^\infty \left(2\acf_X(\tau)-\acf(\tau+\Delta)-\acf(\tau-\Delta)\right) \exp\{-i\tau\omega\}\de\tau\\
    &=2f_X(\omega)-e^{i\omega\Delta}f_X(\omega)-e^{-i\omega\Delta}f_X(\omega)\\
    &=2(1-\cos(\omega\Delta))f_X(\omega)\\
    &=4\sin^2\left(\frac{\omega\Delta}{2}\right)f_X(\omega).\label{eq:differencedsdf}
\end{align}
Therefore, differencing can be easily incorporated into the fitting techniques that have been discussed in this paper, by simply replacing $X$ with the differenced process $Y$ and $f_X$ with $f_Y$ using the relation given by \eqref{eq:differencedsdf}.
Consider the correlation matrix of the periodogram: the matrix with $i,j$th element defined to be the correlation between the periodogram at the $i$th and $j$th Fourier frequencies.
Figure \ref{fig:correlation} shows a plot of the correlation matrix for the periodogram of a wind-generated wave process, estimated using the technique described in Appendix \ref{append:var}.
We can see that for data recorded at a 1Hz sampling rate there is little correlation in the periodogram; however, this is not the case for 4Hz data.
We can also see that the periodogram of the differenced process is almost completely uncorrelated, even for the 4Hz data.\footnote{It should be noted that the region of high correlation in the top left corner of each of the correlation matrices is part of the reason for removing such frequencies from the objective function when performing fits, as discussed further in Section~\ref{sec:freq_selection}.}
From the signal processing perspective, this has reduced the dynamic range of the spectrum as we are multiplying the spectral density function by something that is close to zero for angular frequencies that are small, but is close to one near the Nyquist, down-weighting the peak far more than the tail.
\begin{figure}[htp!]
    \centering
    \begin{subfigure}[t]{0.475\textwidth}
        \centering
        \includegraphics{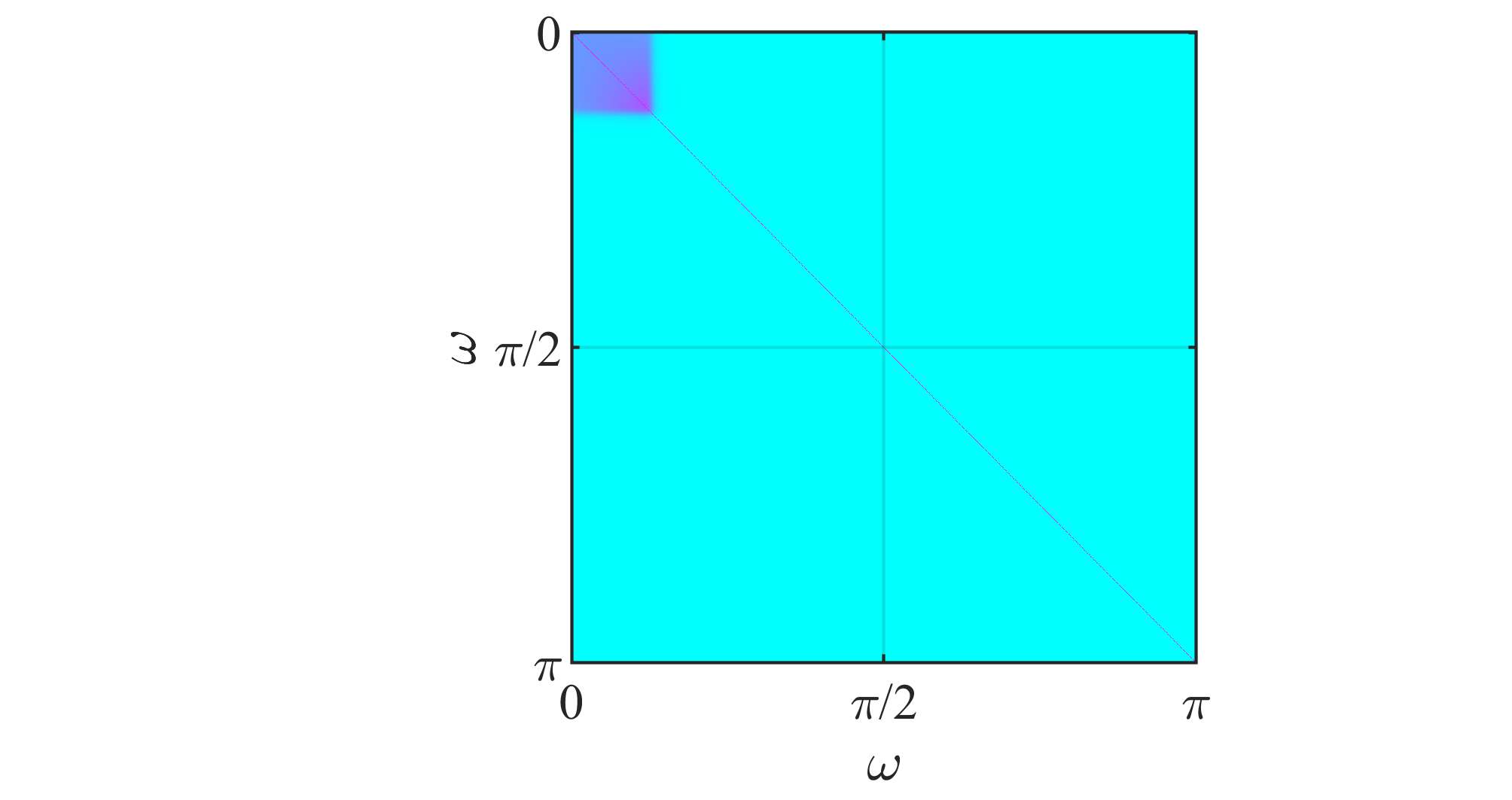}
        \caption{No differencing and $\Delta=1$.}
    \end{subfigure}
    \hfill
    \begin{subfigure}[t]{0.475\textwidth}
        \centering
        \includegraphics{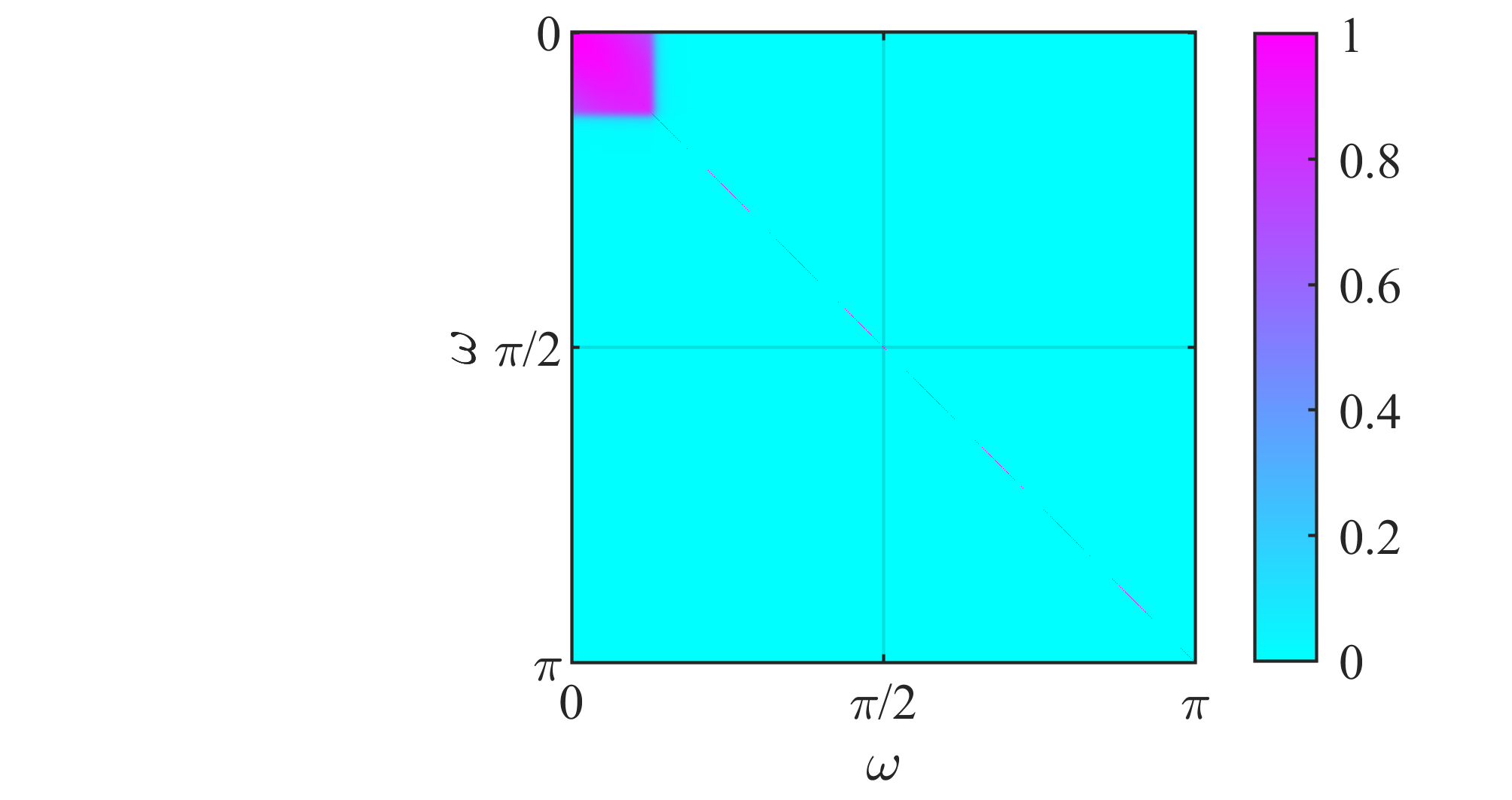}
        \caption{Differencing and $\Delta=1$.}
    \end{subfigure}%
    \vskip\baselineskip
    \begin{subfigure}[t]{0.475\textwidth}
        \centering
        \includegraphics{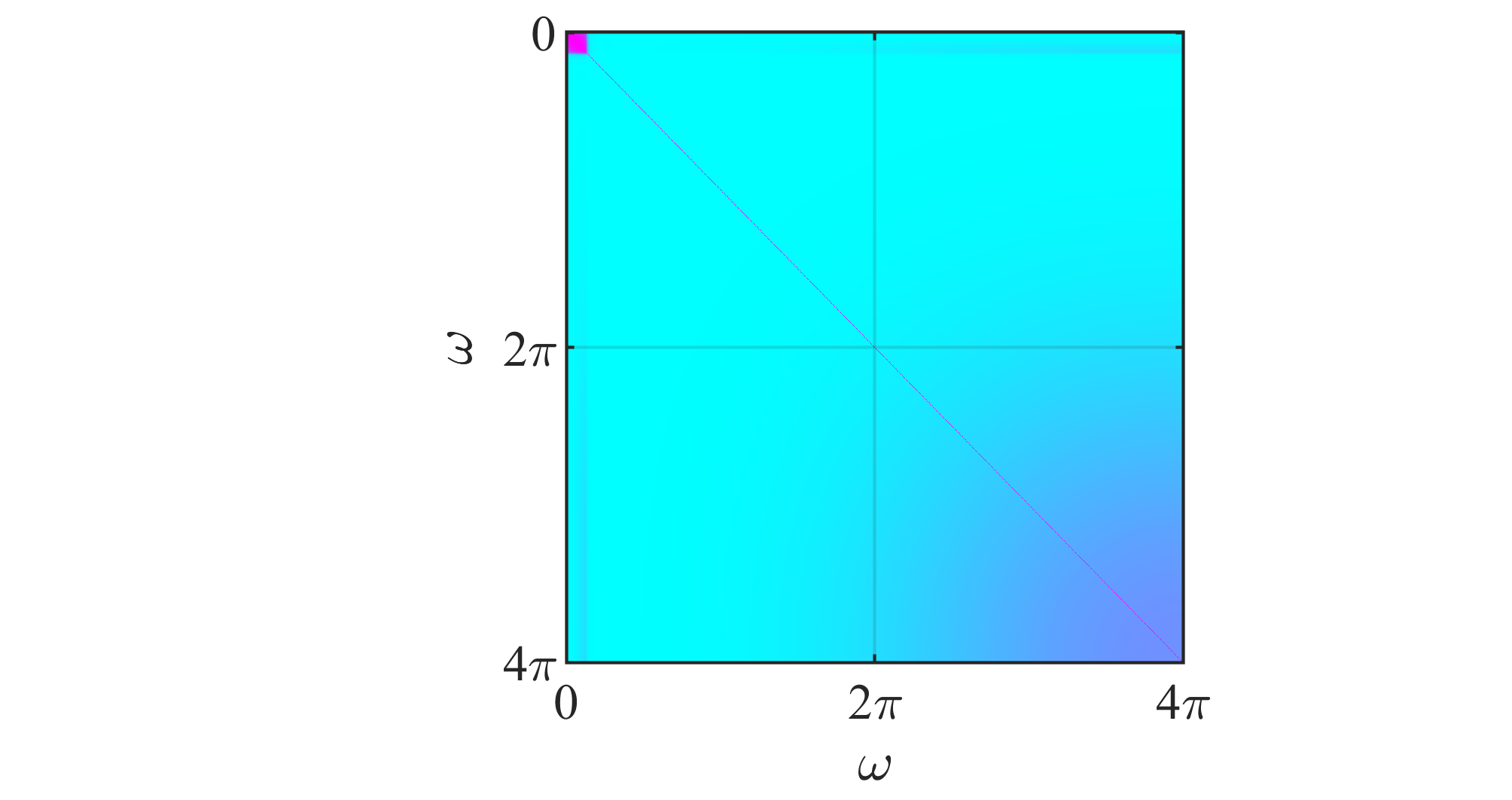}
        \caption{No differencing and $\Delta=1/4$.}
    \end{subfigure}
    \hfill
    \begin{subfigure}[t]{0.475\textwidth}
        \centering
        \includegraphics{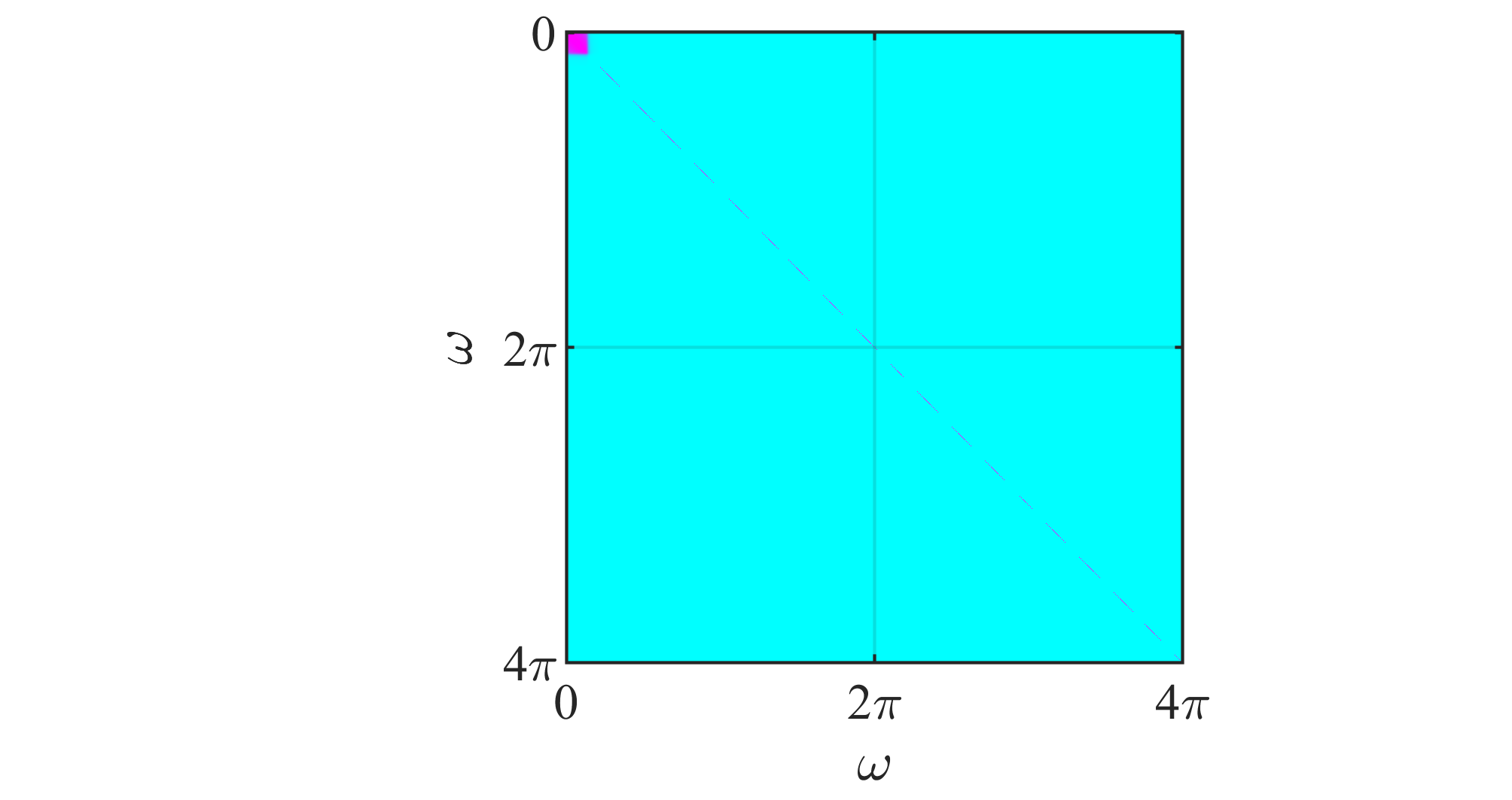}
        \caption{Differencing and $\Delta=1/4$.}
    \end{subfigure}%
    \caption{Image plots of the correlation matrix of the periodogram of a wind-generated wave process for different values of $\Delta$ and the corresponding images for the differenced process, for a generalised JONSWAP with parameters $\alpha=0.7$, $\omega_p=0.7$, $\gamma=3.3$ and $r=5$.}
    \label{fig:correlation}
\end{figure}
    
\section{Practical concerns for implementation with the generalised JONSWAP} \label{sec:practical}
    In Section  \ref{sec:fitting}, we described some techniques that can be used to estimate model parameters. When implementing these techniques for ocean wave models, there are some practical concerns that must be addressed. Firstly, we need not use all of the Fourier frequencies when fitting the model. Indeed, it may be preferable to remove some frequencies that are contaminated by some other process or by observational noise.
    Secondly, there is no known analytical form for the autocovariance corresponding to many of the spectral density functions used when modelling ocean waves. Therefore, numerical techniques for estimating the autocovariance play an important role in many of the fitting procedures discussed in Section  \ref{sec:fitting}.
    In particular, it is necessary for both the de-biased Whittle likelihood and for full maximum likelihood.
    \subsection{Frequency selection}\label{sec:freq_selection}
    Many of the estimators defined in Section  \ref{sec:fitting} involve minimising or maximising objective functions, which are expressed as the sum over some set of frequencies $\Omega\subseteq\Omega_{N,\Delta}$.
    The most simple choice for this set $\Omega$ is just the set of Fourier frequencies $\Omega_{N,\Delta}$. At first glance, this would seem like the most sensible choice (as omitting frequencies is essentially the same as throwing away data-points).
    However, there are many different circumstances in which it is preferable to remove some of the frequencies from the fit.
    \par
    One practical reason for removing certain frequencies is that for very low frequencies, the generalised JONSWAP spectra is zero to machine precision. This often introduces numerical instabilities, especially for objective functions that involve dividing by the spectral density function (such as the Whittle likelihood).
    As can be seen in Figure~\ref{fig:correlation}, there is also a region of high correlation in the low frequencies, which provides an additional motivation for removing such frequencies. An alternative method to reducing correlations in the periodogram is to use tapered versions of the spectral density estimate in the Whittle likelihood \citep{Dahlhaus1988}, but in simulations (available on GitHub) we found omitting frequencies from the fit to be a better solution than tapering in terms of the resulting bias and variance of parameter estimates.
    Another reason for removing certain frequencies from the fit is that it can help to remove noise processes that are present in a record. For example, wave records often contain a low-frequency swell component, but we are interested in the parameters of the wind sea component. By removing frequencies in which the swell is dominant, we are better able to model the wind sea component of a sea state.
    On top of this, there is an added technical concern when using the Whittle and de-biased Whittle likelihoods. The zero and Nyquist frequencies must be omitted (or a modified version of the summand must used for those frequencies). This is because these methods are based on the asymptotic distribution of the periodogram, which is different at the 
    Nyquist and zero frequency than it is at other frequencies.
    \par
    Fitting the model in this way can be thought of as fitting a semi-parametric model (as some of the frequencies are being modelled using a parametric model, and the remaining frequencies by some non-parametric model such as the periodogram). It is worth noting that this approach can actually be applied to full maximum likelihood as well. This can be achieved by transforming both the observations and autocovariance of the model into the frequency domain, applying a band pass filter, and then transforming back. While this is possible in theory, it is fiddly in practice and is no longer exact. This demonstrates another major advantage of spectral-likelihoods: it is far easier to filter out undesired frequencies from the model fit.
    However, the choice of frequencies to be used in the fit should be made prior to the objective function being optimised. 
    Otherwise the number of degrees of freedom could be changing throughout the optimisation routine, which would likely result in additional bias.
    \subsection{Numerical estimation of the autocovariance}
    To calculate both the multivariate Gaussian likelihood and de-biased Whittle likelihood we require the autocovariance of the process given a certain parameter choice (in \eqref{eq:ML} for the multivariate Gaussian likelihood and \eqref{eq:EI} for the de-biased Whittle likelihood). For the generalised JONSWAP spectra, there is no analytical form for this autocovariance. As such, the autocovariance must be approximated numerically.
    Firstly, recall that the autocovariance is the Fourier transform of the spectral density function, as defined in~\eqref{eq:acf:continuous}, and we wish to obtain the autocovariance at lags $0,\Delta,\ldots,(N-1)\Delta$. The first problem we encounter is that this integral is over the entire real line. Clearly, we cannot approximate such an integral numerically and must instead settle for integrating up to some finite frequency, such that the spectral density function beyond that frequency is sufficiently small. In particular, it is convenient to choose a multiple of the Nyquist frequency, as the integral will be approximated using a Fast Fourier Transform, so the desired lags can be extracted by sub-sampling if a multiple of the Nyquist is used in the integration. Therefore, based on equation \eqref{eq:acf:continuous}, we can construct the approximate autocovariance
    \begin{align}
        \hat\acf(\tau)&=\int_{-L\pi/\Delta}^{L\pi/\Delta} f(\omega)\exp\{i\tau\omega\} \de \omega,\label{eq:hat_acf}
    \end{align}
    for $L\in\NN=\{1,2,3,\ldots\}$.
    \par
    Alternatively, we could consider the relation given in equation \eqref{eq:acf:discrete}, between the autocovariance and the discrete time spectral density function. In this case, we would first need to approximate the spectral density function for the discrete time process. To do this, we use a truncated version of the relation given by equation \eqref{eq:sdf:aliasing}, between the spectral density of the continuous and discrete time processes. Consider the approximation of the aliased spectral density function given by
    \begin{align}
        \tilde f_\Delta(\omega)&=\sum_{k=-K}^{K} f\left(\omega+\frac{2\pi k}{\Delta}\right),
    \end{align}
    for $K\in\NN_0=\{0,1,2,\ldots\}$. Then we can construct the alternative approximation to the autocovariance given by
    \begin{align}
        \tilde\acf(\tau)&=\int_{-\pi/\Delta}^{\pi/\Delta} \tilde f_\Delta(\omega)\exp\{i\tau\omega\}\de\omega.\label{eq:tilde_acf}
    \end{align}
    Notice that we may write
    \begin{align}
        \tilde\acf(\tau)&=\int_{-\pi/\Delta}^{\pi/\Delta} \sum_{k=-K}^K f\left(\omega+\frac{2\pi k}{\Delta}\right)\exp\{i\tau\omega\}\de\omega\\
        &=\sum_{k=-K}^K\int_{-\pi/\Delta}^{\pi/\Delta} f\left(\omega+\frac{2\pi k}{\Delta}\right)\exp\{i\tau\omega\}\de\omega\\
        &=\sum_{k=-K}^K\int_{(2k-1)\pi/\Delta}^{(2k+1)\pi/\Delta} f\left(\omega\right)\exp\{i\tau\omega\}\de\omega\\
        &=\int_{(-2K-1)\pi/\Delta}^{(2K+1)\pi/\Delta} f\left(\omega\right)\exp\{i\tau\omega\}\de\omega.\label{eq:equiv_acf}
    \end{align}
    From \eqref{eq:equiv_acf} we can see that, if $L=2K+1$, then $\hat\acf(\tau)$ and $\tilde\acf(\tau)$ are equivalent. 
    In practice, these integrals must be approximated numerically. 
    To do this, we consider a Riemann approximation with bins of width $2\pi/M\Delta$. 
    By choosing $M$ to be some integer bigger than $2N$, we can obtain the desired lags by performing a Fast Fourier Transform and then sub-sampling appropriately. 
    \par
    We can now see that the approximation based on $\hat\acf(\tau)$ (in \eqref{eq:hat_acf}) can be computed in $O(LM\log LM)$ time. However, computing the second approximation, based on $\tilde\acf(\tau)$ (in \eqref{eq:tilde_acf}), requires first computing $\tilde f(\omega)$ at $M$ frequencies (taking $O(LM)$ operations) and then performing a Fourier transform on $M$ frequencies, requiring $O(M\log M)$ operations.
    In other words, the first approach requires $O(M(L\log L+L\log M))$ operations, whereas the second only requires $O(M(L+\log M))$ operations.
    For this reason, we use the latter approach when approximating the autocovariance: first approximating the aliased spectral density, then approximating the autocovariance. 
    \par
    The choice of $K$ (or equivalently $L$) depends on the tail behaviour of the spectral density function in question. In practice, we choose $K$ so that for frequencies beyond $(2K+1)\pi/\Delta$, the spectral density is below some threshold (typically $10^{-6}~ \text{m}^2~\text{s}~\text{rad}^{-1}$), though $K$ should really be chosen so that it scales with $N$, for convergence results to still apply.
    The choice of $M$ is based on the required accuracy of the integral approximation and should be tuned accordingly.
    For the generalised JONSWAP, we have found that $M=\max\{8192, 2N\}$ is a good choice.
\section{Simulation Study} \label{sec:sim}
Though it is possible to make theoretical statements about the asymptotic behaviour of different estimators, from a practical perspective, their finite sample behaviour is of primary interest.
To investigate this, we perform a simulation study to assess the performance of the estimators described in Section  \ref{sec:fitting}.
In this simulation study, we compare six different fitting techniques based on these estimators. The first, which we call least squares, uses the curve fitting approach with the periodogram. The second approach is similar, but uses Bartlett's method to estimate the spectral density function, which we refer to as Bartlett least squares. The window size is chosen so that we have a spectral resolution of $0.2\pi$, i.e. the window size is $100/\Delta$.\footnote{Clearly for some values of $\Delta$ this would not be an integer; however, for the values of $\Delta$ that we choose it is.} For 1.28Hz data, this corresponds to a window size of 128. The third and fourth approaches are the Whittle and aliased Whittle likelihoods respectively. The final two approaches are the de-biased Whittle likelihood and full time domain maximum likelihood.
\subsection{Method}
To investigate the effectiveness of different fitting approaches we simulate a linear wave record with a known parametric spectral density function and then re-estimate the parameters from the simulated record. By repeatedly performing this process, we can assess the bias and variance of each of the estimators discussed in Section  \ref{sec:fitting}.
For the purposes of the simulation study we let $\XN$ be a random variable with a multivariate normal distribution resulting from sub-sampling the continuous-time mean-zero stationary Gaussian process $\bX$, where $\bX$ has spectral density function $f_G(\omega\mid\theta)$, defined by \eqref{eq:genJONSWAP}. We then simulate a realisation of $\XN$ using the circulant embedding method described by \cite{Davies1987} (and for complex valued processes by \cite{Percival2006}).
We choose to use circulant embedding over the typical approaches for simulating Gaussian processes often used in the ocean waves literature, such as the method due to \cite{Tucker1984}, as these methods only approximately simulate a Gaussian process with the given spectral density function, whereas circulant embedding is exact (up to the quality of the approximation of the autocovariance that is used). Furthermore, many techniques, such as the method proposed by \cite{Tucker1984}, or the more recent modification due to \cite{Merigaud2018} do not account for aliasing when simulating the process. 
Since we are explicitly interested in the effect that aliasing has on recovered parameters, it is important that we simulate something that is as close as we can get to a Gaussian process with the desired aliased spectral density function.
Circulant embedding generates time series with all of the sampling effects discussed by \cite{Tucker1984}, but also includes additional finite sampling effects such as aliasing and correlations between spectral estimates at different frequencies. Such effects should be present in generated time series, but are not in time series generated using the method suggested by \cite{Tucker1984}. More details can be found in \cite{Davies1987, Dietrich1997, Wood1994}.
\par
To perform the fitting we first choose one of the objective functions described in Section~\ref{sec:fitting} and optimise this using the fmincon function in \matlab (with maximisation done by minimising the negative of the objective function). An initial guess for the fitting procedure needs to be provided for each of the parameters. For $\omega_p$, we use the frequency corresponding to the largest value of the periodogram. For $r$, we use a basic linear regression coefficient between the log spectral density and log periodogram over the tail frequencies (where the tail is chosen to be all frequencies that are closer to the Nyquist than the peak).
We choose to initialise $\gamma$ by setting it equal to 3. This is because choosing $\gamma$ heuristically is not easy, and $\gamma=3$ is close to the value commonly assumed by many oceanographers. In practice, the initial choice of $\gamma$ does not seem to have a huge impact on the final fitted values; however, the optimisation could also be run with multiple starting values of $\gamma$ and the best estimate could then be selected. Once these parameters are initialised, $\alpha$ is initialised so that the area under the initial parametric spectral density function matches the area under the periodogram.
In simulations, we find that the inference is not sensitive to the initial guess (provided it is sensible).
In practice we are often fitting models to multiple consecutive sea states. In this case, it can be more efficient to use the parameter estimates for the previous sea state as initial values when optimising.

\subsection{A canonical sea state}\label{sec:single_sea_state}
We shall begin by considering how each of the estimators perform for one choice of true parameters, before showing that the results are robust to the true parameters. In particular, we begin by considering a spectral density function of the form described in Section  \ref{sec:ocean_model}, with $\sigma_1=0.07$, $\sigma_2=0.09$ and $s=4$ treated as known, and with $\alpha=0.7$, $\omega_p=0.7$, $\gamma=3.3$ and $r=4$ treated as unknown parameters to be estimated.
The reason for choosing these parameters is that $\alpha=0.7$ roughly corresponds to the scaling present when using Phillip's constant in a JONSWAP spectra, $\omega_p=0.7$ is a reasonable choice for peak frequency, $\gamma=3.3$ is commonly assumed to be the peak enhancement factor, and $r=4$ is one of the suggested values for the tail decay index.
Half hour records sampled at 1.28Hz (a standard time interval and sampling frequency for wave records) were simulated and the parameters were estimated using each of the six estimation methods described above. The resulting estimates across 1000 repeated simulations are summarised in Figure \ref{fig:half_hour}, alongside the time taken to perform the optimisation. For comparison, the true value of each parameter is given by a horizontal black dashed line. 
\par
Perhaps the most striking feature of Figure \ref{fig:half_hour} is the difference in the variability in estimates of the tail parameter, $r$, when comparing least squares type techniques to likelihood based techniques.
Least squares techniques recover parameter estimates ranging from well beyond three to five, making it very difficult to make any statements about the true value of the tail decay.
However, we can see that statistical techniques such as the de-biased Whittle likelihood are able to recover the original tail parameter to within a few decimal places.
Therefore, by using the de-biased Whittle likelihood, practitioners would be able to distinguish between $\omega^{-4}$ and $\omega^{-5}$ spectral tails in observed records. 
Though it should be noted that this assumes the wind-sea had a spectral density that is well described by a generalised JONSWAP and so we cannot provide model free estimates of the tail decay. In particular, a different model would be required for a transition in tail decays over frequency, such as the effect discussed by \cite{Babanin2010}; however, the de-biased Whittle likelihood is generic and could be applied to other models provided they satisfied certain conditions (see Appendix~\ref{append:assumption}).
We can also see that the de-biased Whittle likelihood offers an improvement in estimates of $\gamma$, performing almost as well as full maximum likelihood.
\par
It is also interesting that bias can be seen in both the Whittle and aliased Whittle likelihood estimates, but that this bias is not present in the de-biased Whittle likelihood estimates.
This verifies that the de-biased Whittle likelihood is indeed accounting for some of the bias present in standard Whittle likelihood, and demonstrates why de-biased Whittle likelihood is necessary over the aliased Whittle likelihood, which can still be seen to be biased for some parameters.
In the analysis of Figure \ref{fig:half_hour} full maximum likelihood provides the best estimates, in terms of root mean square error.
However, this comes at significant computational cost, whilst giving limited improvement in bias and variance when compared to the de-biased Whittle likelihood.
\par
\begin{figure}[htp]
    \centering
    \includegraphics{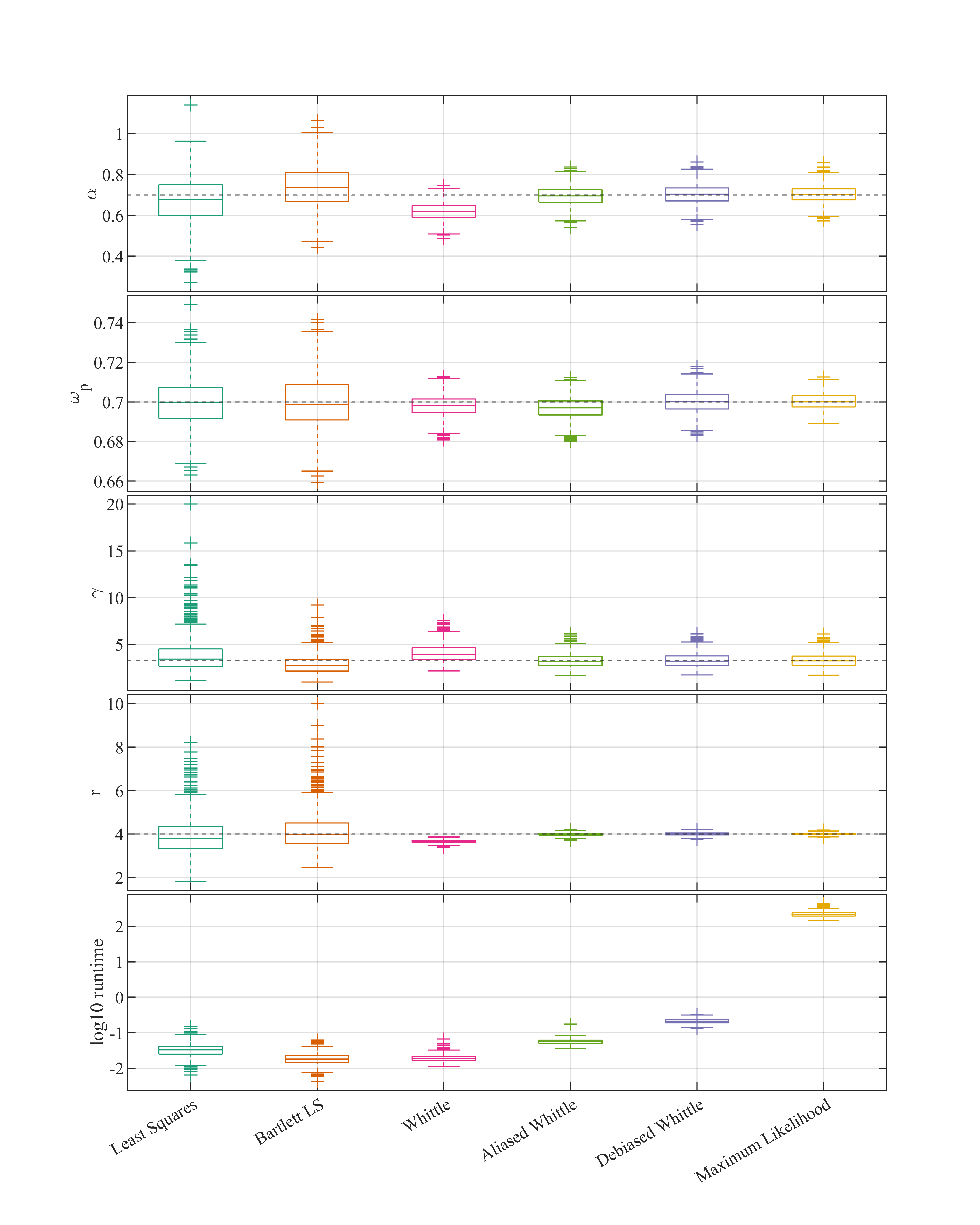}
    \caption{Boxplots of parameter estimates and time taken for six different fitting routines with true parameters denoted by the dashed black lines. Each row displays the results for a given fitting technique, as well as the log of the time taken to perform the optimisation, recorded in seconds. The fitting was performed on simulations of 1.28Hz observations recorded for half an hour (2304 observations). The process was repeated 1000 times.}
    \label{fig:half_hour}
\end{figure}
\par
Often, during optimisation, parameters may trade off against one-another. Therefore it is also important to look at the joint behaviour of parameter estimates.
Figure \ref{fig:half_hour_scatter} shows a scatter plot of the de-biased Whittle likelihood estimates from Figure \ref{fig:half_hour}. We can see that there is very strong correlation between the estimates of $\alpha$ and $r$, and some negative correlation between $\alpha$ and $\gamma$. This likely occurs because $\gamma$ and $r$ change the area under the spectral density function, so $\alpha$ is likely to be adjusted to compensate. Though it would be possible to reparameterise to try and avoid this, it does not seem to have a significant impact on the resulting estimates and is therefore not necessary.
\par
In practice, longer sea states are often used to estimate model parameters. Therefore, we also compare some of the methods for 3 hour records. Figure \ref{fig:three_hour} shows the comparison of least squares, Bartlett least squares and de-biased Whittle likelihood estimates for these 3 hour records. The variance in the first two estimators has indeed decreased when compared to the estimates from half hour records shown in Figure \ref{fig:half_hour}. 
However, by comparing the de-biased Whittle likelihood estimates in Figure \ref{fig:half_hour} to the least squares estimates in Figure \ref{fig:three_hour}, we can see that the de-biased Whittle likelihood used on a half hour record yields better estimates than the least squares based estimates performed on 3 hour records.
The longer record reduces the variance of the least squares and Bartlett least squares techniques enough to allow us to see another interesting feature, namely that there is significant bias present in the Bartlett least squares estimates that is not present in the standard least squares estimates. This demonstrates that non-parametric smoothing can have unexpected consequences when used to fit a parametric spectral density function.
\par
In essence, by using the de-biased Whittle likelihood, we can obtain more accurate and precise parameter estimates, whilst simultaneously reducing the length of the record required to obtain the estimates.
This improvement is especially noticeable (and important) for the peak enhancement factor, $\gamma$, and tail decay index, $r$.
In practice, this has two important consequences. Firstly, we can reduce the amount of time for which the surface is assumed to be stationary. This means that we can fit stationary models to weather systems that evolve very quickly, such as tropical cyclones. 
Secondly, we can obtain parameter estimates at more frequent time intervals. This higher parameter resolution means that we gain a more detailed insight into how certain parameters evolve throughout the course of a meteorological event.
\begin{figure}[htp!]
    \centering
    \includegraphics{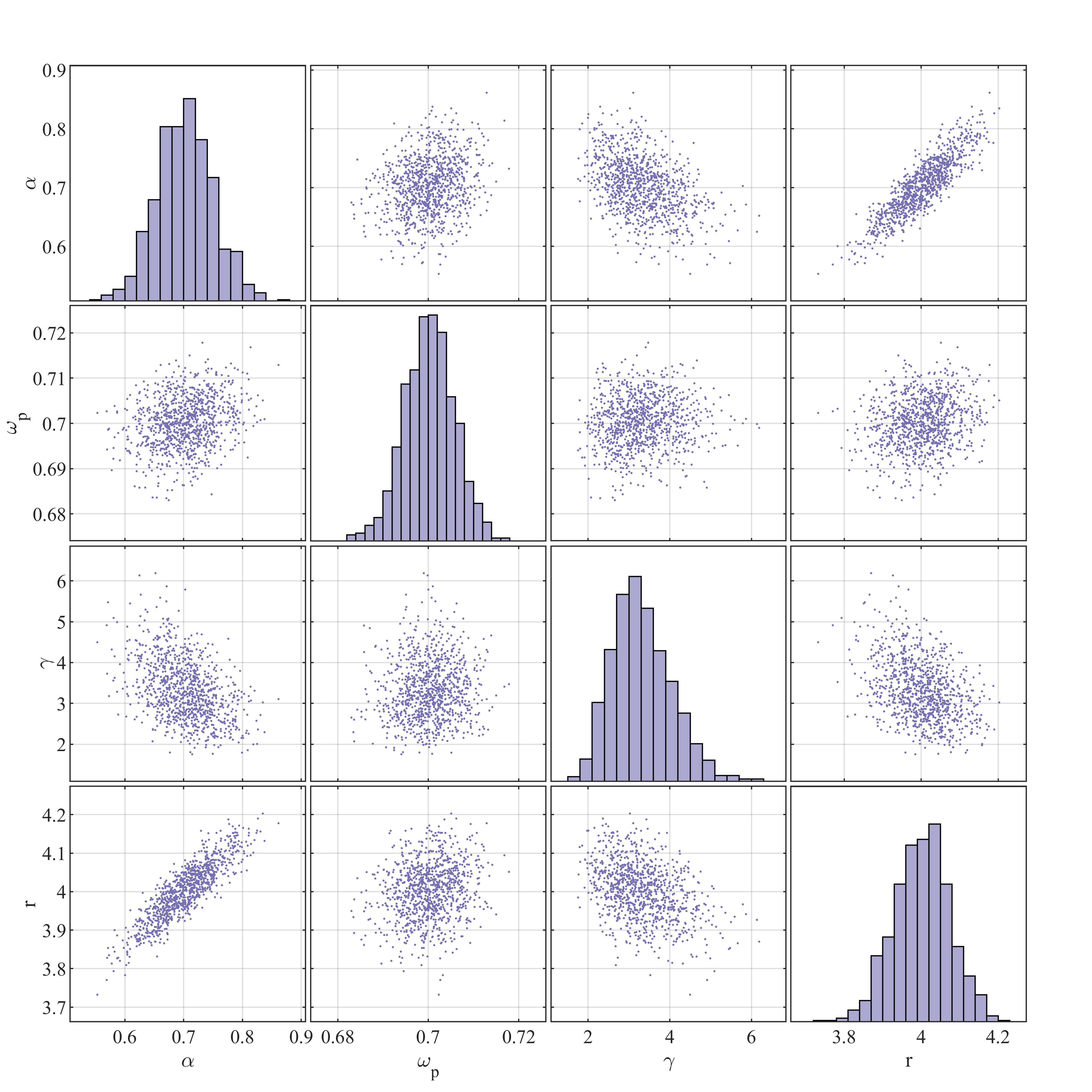}
    \caption{Scatter plot of the de-biased Whittle estimates obtained from half hour records shown in Figure \ref{fig:half_hour}.}
    \label{fig:half_hour_scatter}
\end{figure}
\begin{figure}[htp!]
    \centering
    \includegraphics{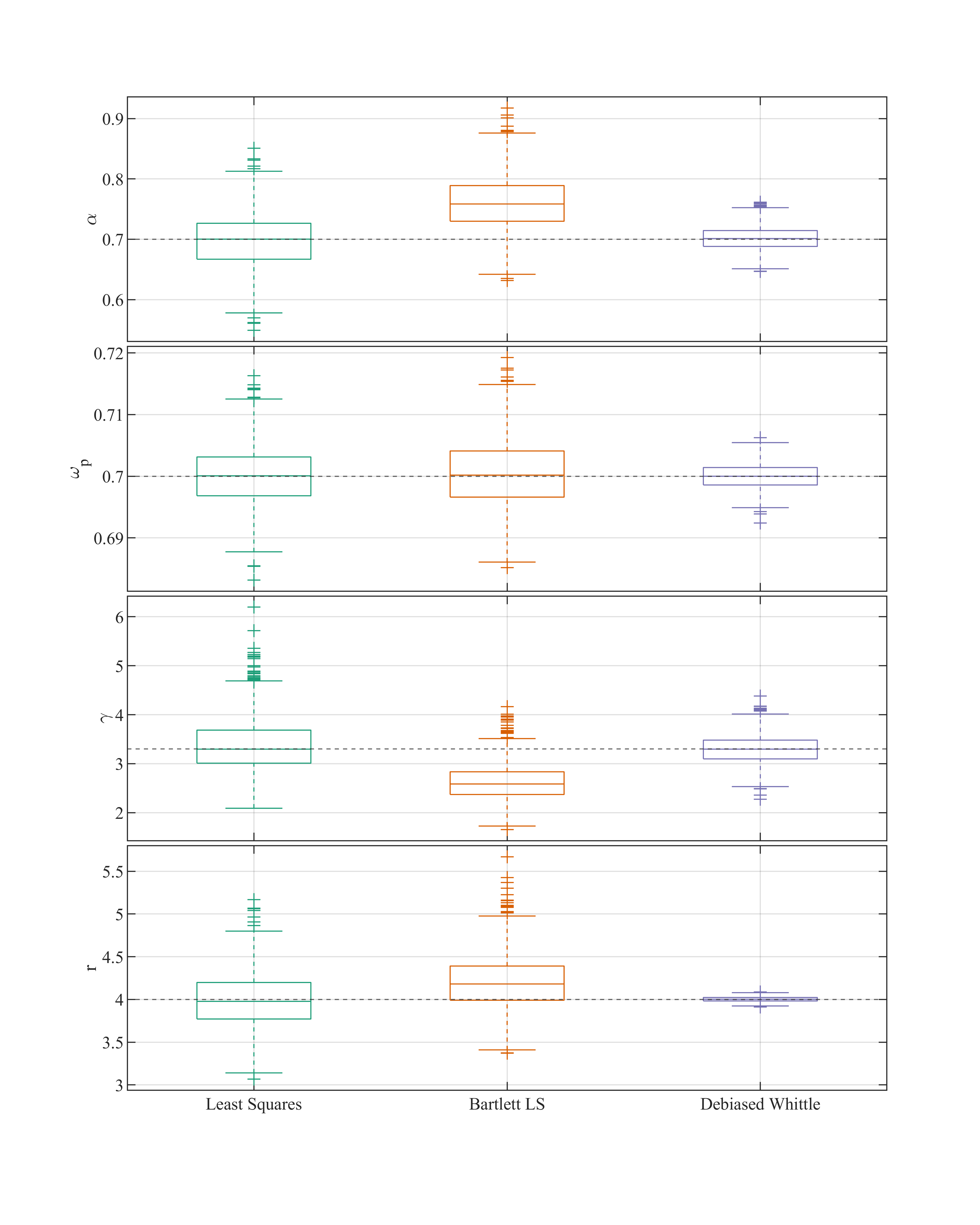}
    \caption{Boxplots of 1000 parameter estimates using three different fitting routines on the same simulations using the same true parameters as Figure \ref{fig:half_hour} and same sampling interval but simulating 3 hour records.}
    \label{fig:three_hour}
\end{figure}
\subsection{Robustness of results}\label{sec:comparison:robust}
In Section  \ref{sec:single_sea_state} we have demonstrated that the de-biased Whittle likelihood can produce parameter estimates that are both more accurate and more precise than those produced by least squares techniques, without making huge sacrifices in terms of computational time. However, it is also important to check that these results extend to different choices of the true parameter.

In Table~\ref{tab:rmse}, we present the results of a simulation study comparing least squares, Bartlett least squares and the de-biased Whittle likelihood over 24 different choices of true parameter. For each choice of true parameter, we calculated the percentage bias, standard deviation (SD) and root mean square error (RMSE) of the estimates (relative to the true parameter). We then averaged over all the choices of true parameter, yielding an average percent bias, SD and RMSE for each parameter.
The parameter choices used to perform these simulations were all the combinations of $\alpha=0.7$; $\omega_p=0.7, 0.9, 1.2$; $\gamma=1,2,3.3,5$; and $r = 4,5$.
Boxplots for each set of parameter choices are available on GitHub \citep{Grainger2021a}.
We can see from Table~\ref{tab:rmse} that there are substantial reductions in both the bias and standard deviation of all the estimated parameters.

\begin{table}[h!]

\begin{tabular}{c|rrr|rrr|rrr}
   & \multicolumn{3}{c}{Bias}  & \multicolumn{3}{c}{SD} & \multicolumn{3}{c}{RMSE}   \\
   & \multicolumn{1}{c}{LS}    & \multicolumn{1}{c}{BLS}   & \multicolumn{1}{c}{DW}   & \multicolumn{1}{c}{LS}    & \multicolumn{1}{c}{BLS}   & \multicolumn{1}{c}{DW} & \multicolumn{1}{c}{LS}    & \multicolumn{1}{c}{BLS}   & \multicolumn{1}{c}{DW} \\ \hline
$\alpha$   & 4.70\%  & 10.88\% & 0.80\% & 34.64\% & 30.69\% & 9.13\%  & 35.46\% & 33.49\% & 9.19\%  \\
$\omega_p$ & 0.08\%  & 0.11\%  & 0.06\% & 1.73\%  & 1.85\%  & 0.77\%  & 1.74\%  & 1.86\%  & 0.77\%  \\
$\gamma$   & 15.88\% & 10.20\% & 2.98\% & 43.29\% & 29.27\% & 18.00\% & 46.39\% & 31.75\% & 18.71\% \\
$r$        & 3.89\%  & 5.60\%  & 0.18\% & 17.57\% & 15.76\% & 2.09\%  & 18.35\% & 17.44\% & 2.11\%  \\\hline
average    & 6.14\%  & 6.70\%  & 1.01\% & 24.31\% & 19.39\% & 7.50\%  & 25.49\% & 21.13\% & 7.69\% 
\end{tabular}

\caption{Average percentage bias, standard deviation and root mean squared error across 24 different parameter choices calculated from 1000 repetitions per parameter choice for least squares (LS), Bartlett least squares (BLS) and de-biased Whittle likelihood (DW). Parameter choices were all the combinations of $\alpha=0.7$; $\omega_p=0.7, 0.9, 1.2$; $\gamma=1,2,3.3,5$ and $r = 4,5$. Simulated records were half an hour long and sampled at 1.28Hz. The bottom row shows the average of these quantities over all the parameters in the model.}
\label{tab:rmse}

\end{table}
\par
True parameters on the boundary of the parameter space can cause problems when performing parametric estimation. As such, we shall further discuss the results for the special case when $\gamma=1$. 
The fitted parameters for half hour simulated records with true parameter values $\alpha=0.7$, $\omega_p=0.7$, $\gamma=1$ and $r=5$, using the least squares, Bartlett least squares and de-biased Whittle likelihood techniques can be seen in Figure \ref{fig:half_hour_gamma_1}. 
This is an interesting case because not only does the true parameter lie on the boundary of the parameter space, but this value of $\gamma$ corresponds to a Pierson-Moskowitz spectrum, for a fully developed sea. This means that such a value of $\gamma$ could occur in nature, and as such it is important that we can model this case.
The problem is that the theoretical guarantees for an approach such as the de-biased Whittle likelihood rely on the assumption that the true parameter does not lie on the boundary of the parameter space. Since this is not the case for $\gamma=1$, therefore we must take care when dealing with records for which the true value of $\gamma$ may indeed be 1. In particular, we can see from Figure \ref{fig:half_hour_gamma_1} that the de-biased Whittle likelihood estimates for $\gamma$ is not normally distributed, a result that is assumed when constructing the confidence intervals that are described later in Section  \ref{sec:confidence}.
One way to deal with this problem is to fit both a model with $\gamma=1$ fixed and one with $\gamma$ as a free parameter, and then test to see if there is evidence for $\gamma>1$. Such a procedure allows us to circumvent potential issues caused by $\gamma$ lying on the boundary of the parameter space.
This could be performed by using the procedure developed by \cite{Sykulski2017}, adapted to the 1D case.
\begin{figure}[htp!]
    \centering
    \includegraphics{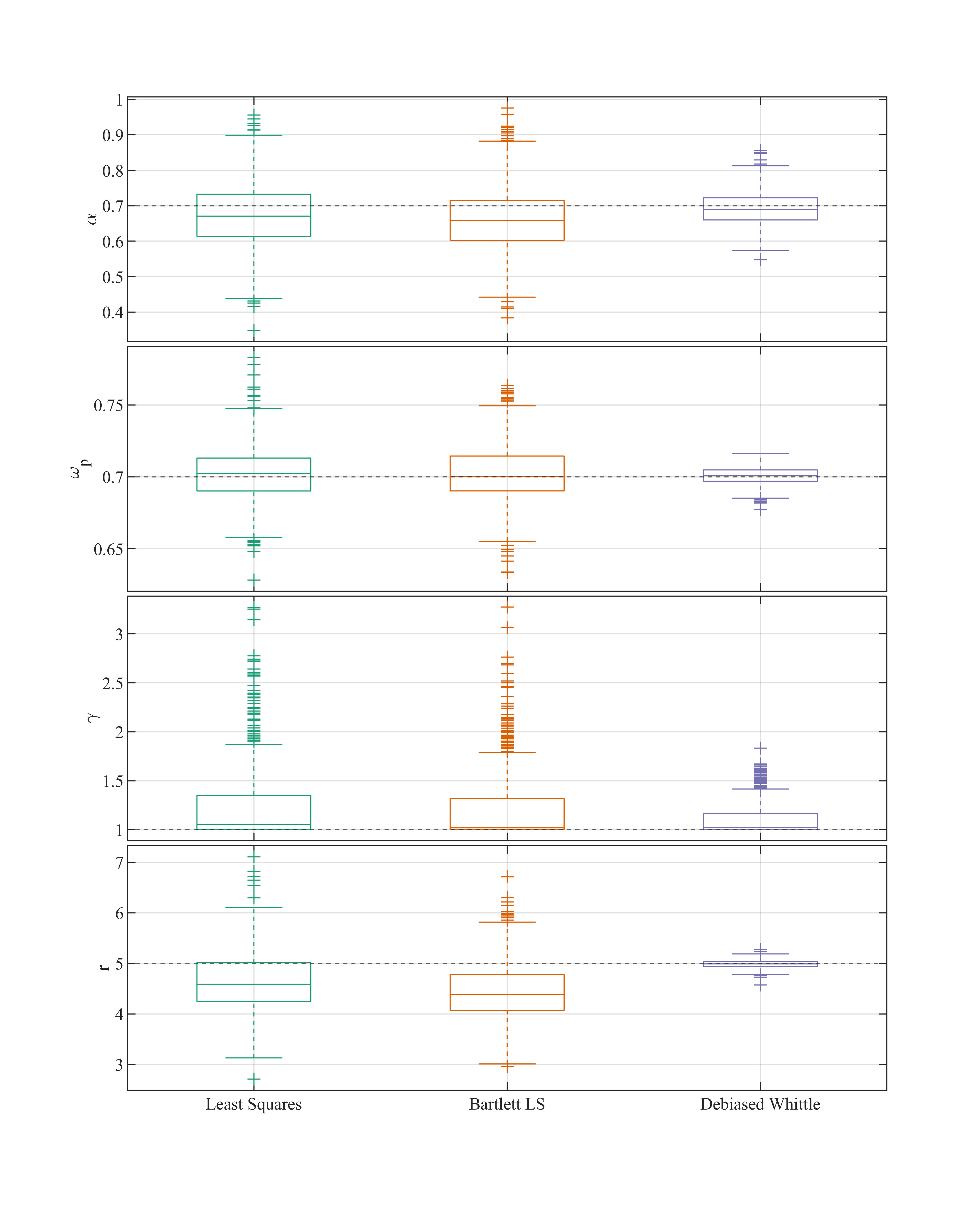}
    \caption{Boxplots of 1000 parameter estimates using three different fitting routines with true parameters $\alpha=0.7$, $\omega_p=0.7$, $\gamma=1$ and $r=5$, estimated from simulated half hour records sampled at 1.28Hz.}
    \label{fig:half_hour_gamma_1}
\end{figure}

\subsection{Differencing for high sampling frequencies}\label{sec:differencing}
As we discussed in Section \ref{sec:fitting:differencing}, for wind generated waves observed at a 4Hz sampling rate, the periodogram is highly correlated. 
As noted in Section \ref{sec:fitting:comparison}, in this case, we would expect spectral techniques to perform poorly compared to full maximum likelihood. 
Differencing can reduce the correlation in the periodogram, and therefore can be a powerful tool to remove bias from spectral methods.
Figure \ref{fig:diffBoxplot} shows box plots of estimated parameters for least squares, Bartlett least squares and de-biased Whittle likelihood techniques, both with and without differencing, for both 1Hz and 4Hz data (in this case, 2048 seconds worth of data was simulated per record, i.e. 2048 observations for the 1Hz simulation and 8192 for the 4Hz simulation). We can see little benefit from differencing in the 1Hz estimates (Figure \ref{fig:diffBoxplot:1}). However, when it comes to 4Hz data, there is a major benefit to differencing, especially for the de-biased Whittle likelihood.
Because there is little difference between the de-biased Whittle fits on the original and differenced 1Hz data, we would recommend that differencing is used as standard, to protect against the issues seen in Figure \ref{fig:diffBoxplot:4}. At the very least, investigating the correlation matrix of the periodogram should be an important diagnostic when fitting spectral models.
\begin{figure}[htp!]
    \centering
    \begin{subfigure}[t]{0.475\textwidth}
        \centering
        \includegraphics{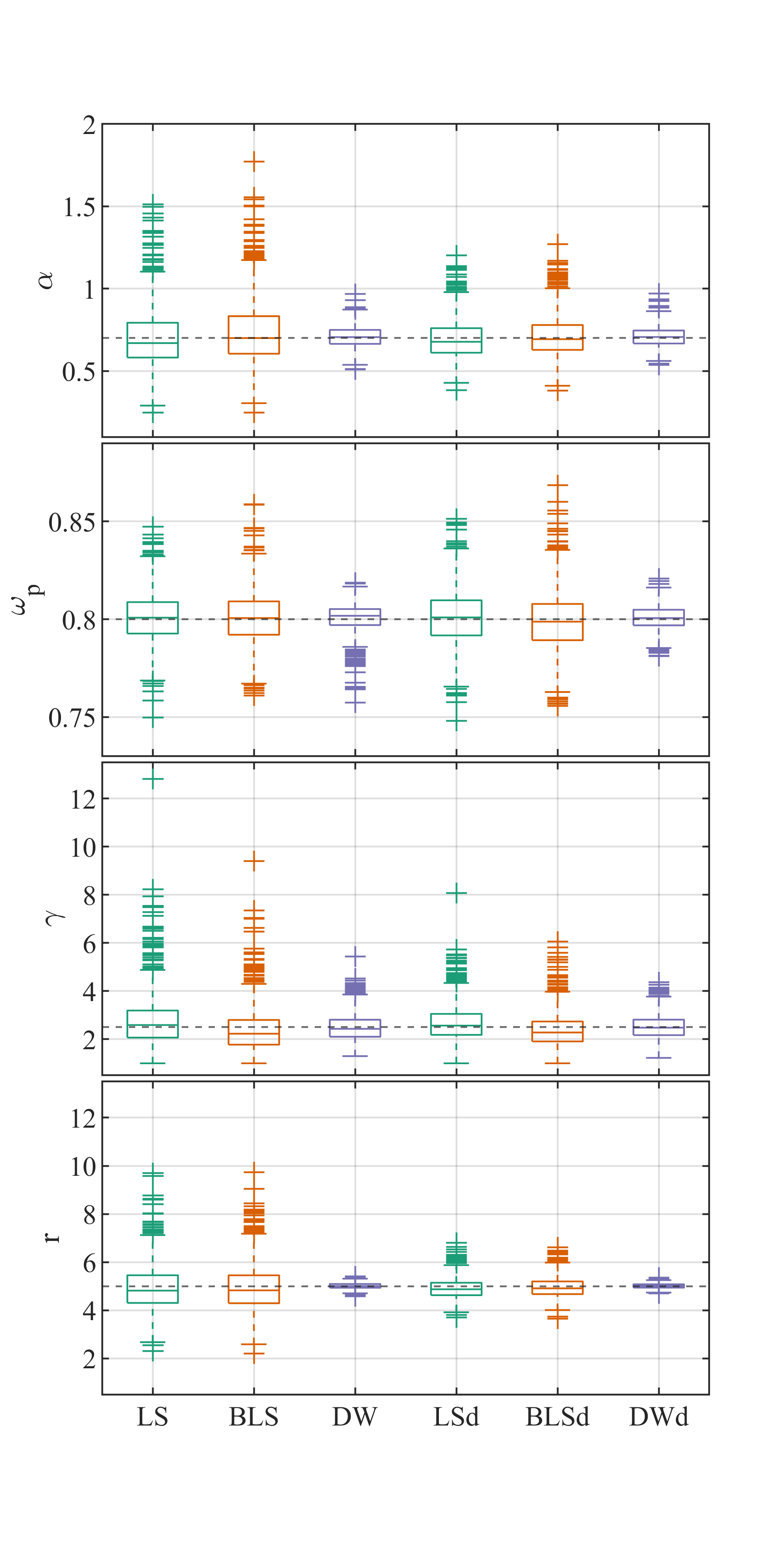}
        \caption{Parameter estimates for 1Hz data.}
        \label{fig:diffBoxplot:1}
    \end{subfigure}
    \hfill
    \begin{subfigure}[t]{0.475\textwidth}
        \centering
        \includegraphics{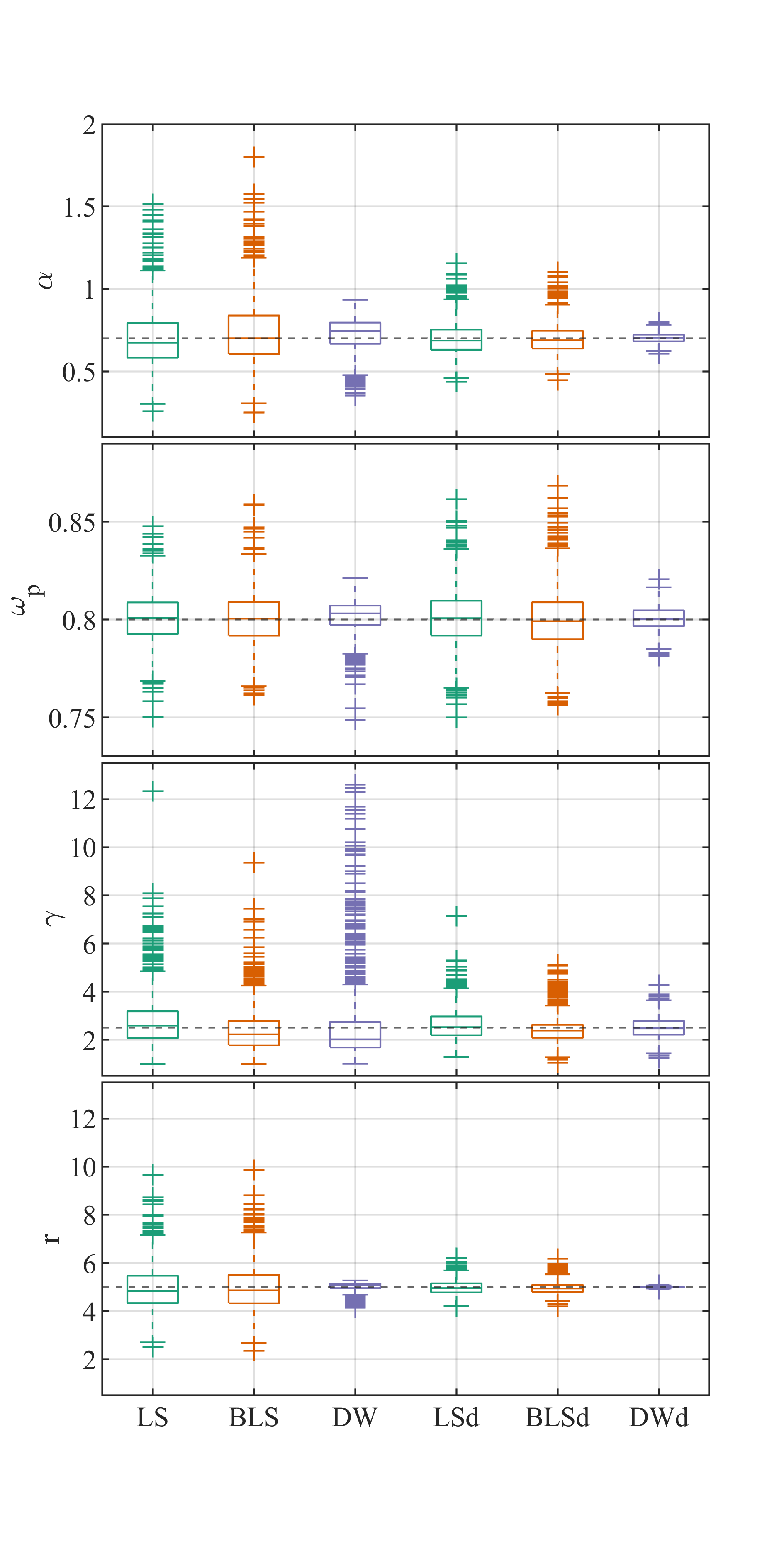}
        \caption{Parameter estimates for 4Hz data.}
        \label{fig:diffBoxplot:4}
    \end{subfigure}%
    \caption{Boxplots of 1000 parameter estimates for 1Hz and 4Hz data. LS, BLS, DW denote least squares, Bartlett least squares and de-biased Whittle respectively, and LSd, BLSd and DWd denote the differenced versions. Note that the values of $\gamma$ for the 4Hz DW have been truncated, the maximum was 19.81.}
    \label{fig:diffBoxplot}
\end{figure}

\section{Quantifying estimation uncertainty}\label{sec:confidence}
As with any statistical analysis, it is important to quantify the uncertainty in the parameter estimates that are obtained using the de-biased Whittle likelihood. \cite{Sykulski2016a} developed an approach for quantifying the asymptotic variance of the estimator for a single time series. Say that we have $p$ parameters that are written in a vector $\theta$ and are interested in the de-biased Whittle likelihood estimator $\hat\theta_{DW}$ of the true parameter vector $\theta$. \cite{Sykulski2016a} decompose the variance as
\begin{align}
    \var{\hat\theta_{DW}}&\approx\EE{H(\theta)}^{-1}\var{\nabla\ell_{DW}(\theta)}\EE{H(\theta)}^{-1},\label{eq:CI:var:dw}
\end{align}
where $H(\theta)$ denotes the matrix of second derivatives of the log-likelihood function $\ell_{DW}(\theta)$, and $\nabla$ denotes the vector of first derivatives of $\ell_{DW}(\theta)$.
\par
Now for the de-biased Whittle likelihood, we may write
\begin{align}
    \frac{\partial}{\partial\theta_j}\ell_D(\theta)&=-\sum_{\omega\in\Omega}\left\{\left(\frac{1}{\bar{f}_N(\omega\mid\theta)}-\frac{I(\omega)}{\bar{f}_N^{2}(\omega\mid\theta)}\right)\frac{\partial}{\partial\theta_j}\bar{f}_N(\omega\mid\theta)\right\},
\end{align}
for $j=1,\ldots,p$. As will be seen in Appendix~\ref{append:var}, this is required to calculate $\var{\nabla\ell_{DW}(\theta)}$. Furthermore, we may write 
\begin{align}
    \EE{\frac{\partial^2}{\partial\theta_j\partial\theta_k}\ell_D(\theta)}&=-\sum_{\omega\in\Omega}\left\{\frac{1}{\bar{f}_N^2(\omega\mid\theta)}\frac{\partial}{\partial\theta_j}\bar{f}_N(\omega\mid\theta)\frac{\partial}{\partial\theta_k}\bar{f}_N(\omega\mid\theta)\right\},
\end{align}
for $j,k=1,\ldots,p$. Therefore, we require only the first derivative of the expected periodogram in order to compute both parts of the variance decomposition given by \eqref{eq:CI:var:dw}.
\par
Furthermore, the triangle function and Fourier basis are constant with respect to $\theta$, so from~\eqref{eq:EI} we have that
\begin{align}
    \frac{\partial}{\partial\theta_j}\bar{f}_N(\omega\mid\theta)&=\frac{1}{2\pi}\text{Re}\left(2\Delta\sum_{\tau=0}^{N-1}(1-\tau/N)\frac{\partial}{\partial\theta_j}\left\{\gamma(\tau\mid\theta)\right\}\exp\{-i\omega\tau\Delta\}-\Delta\frac{\partial}{\partial\theta_j}\left\{\gamma(0\mid\theta)\right\}\right),
\end{align}
for $j=1,\ldots,p$.
In other words, to calculate the derivative of the expected periodogram, we may first calculate the derivative of the autocovariance, then calculate the expected periodogram of the process with that derivative as its autocovariance. In the case of the generalised JONSWAP spectral density function, when an analytical form for such autocovariance is unavailable, we may instead approximate the derivative of the autocovariance by first approximating the derivatives of the aliased spectral density function, and then of the autocovariance.
\par
In Appendix~\ref{append:deriv} we show that, for the generalised JONSWAP form, partial derivatives of $f_\Delta(\omega\mid\theta)$ can be calculated from the derivatives of the spectral density function, and that the resulting derivatives are continuous.
Therefore, by Leibniz's rule
\begin{align}
    \frac{\partial}{\partial\theta_j}\gamma(\tau\mid\theta)&=\int\limits_{-\pi/\Delta}^{\pi/\Delta} \frac{\partial}{\partial\theta_j}f_\Delta(\omega\mid\theta)\exp\{-i\omega\tau\}\de\omega.
\end{align}
As a result, we may approximate the derivative of the autocovariance from the derivative of the aliased spectral density function in the same way that we approximated the autocovariance from the aliased spectral density function in Section~\ref{sec:practical}.
Therefore, each partial derivative can be computed in the same time it would take to do one function evaluation, which is faster than using numerical approximations for the derivative (e.g. finite differencing methods).
\par
In Appendix \ref{append:var}, we discuss a novel procedure for estimating $\var{\nabla\ell_{DW}(\theta)}$ in a computationally efficient manner. Combining this with the approach for obtaining derivatives, we can estimate the variance of the de-biased Whittle likelihood estimator using \eqref{eq:CI:var:dw}. This enables the computation of approximate confidence intervals for our estimates, by using the asymptotic normality of the estimator and standard theory for confidence intervals.

\section{Maui data}\label{sec:application}
Observations of fetch-limited seas made near the Maui-A platform off the coast of New Zealand between November 1986 and November 1987 have been studied previously by \cite{Ewans1986} and \cite{Ewans1998}.
We revisit these observations to demonstrate additional practical concerns when applying the techniques discussed in this paper to observed wave records.
Observations were taken using a Datawell WAVEC buoy at a location with a water depth of 110m.
Data were recorded for 20 minutes at a sampling rate of 1.28Hz, then there is a gap of 10 minutes while data is processed.
Non-parametric spectral density estimates were recorded for each of these intervals, while the raw time series was only retained for the first 20 minutes of a 3 hour segment. In other words, a 20 minute record is available starting at 00:00, 03:00, 06:00 etc. for each day.
\par
For the sake of illustration, we shall analyse the behaviour of spectral parameters throughout a single storm event, from the 2nd to the 7th of November, 1986.\footnote{Note that this is not intended to make general statements about the behaviour of these parameters. Rather, we are illustrating some of the additional practical concerns faced when analysing recorded wave time series.}
As can be seen from Figure \ref{fig:maui:Boxplot}, throughout the first and second days, the conditions become steadily more severe, before calming down over the proceeding days.
Transforming to the frequency domain, we see the presence of persistent swell, which remains fairly constant throughout all of the records. 
Parameters of the wind-sea component of these bi-modal seas are estimated by first removing contaminated frequencies from the objective function, with optimisation proceeding as previously described.\footnote{This can be thought of as first applying a high pass filter, then fitting a uni-modal model to the resulting filtered record.}
Ideally we would aim to fit a bi-modal model that was designed to also describe the swell component.
However, such an approach is beyond the scope of this paper, though we note that preliminary results for such a procedure are encouraging and are discussed further in Section \ref{sec:discussion}.
It was also observed that frequencies beyond $3 \text{ rad s}^{-1}$  exhibited a different behaviour from the rest of the spectral tail. 
We suspect this may be related to the physical response of the buoy to the waves, though we do not know exactly how such behaviour should be modelled. 
Therefore, we also chose to omit those frequencies in the spectral fitting procedure when using each of the fitting procedures.
After filtering out the aforementioned frequencies, a generalised JONSWAP spectral form was fitted to each of the sea states in turn, using least squares, Bartlett least squares, and the de-biased Whittle likelihood. These estimates are shown in Figure \ref{fig:maui:ls}. 
\begin{figure}[htp!]
    \centering
    \begin{subfigure}[t]{0.475\textwidth}
        \centering
        \includegraphics{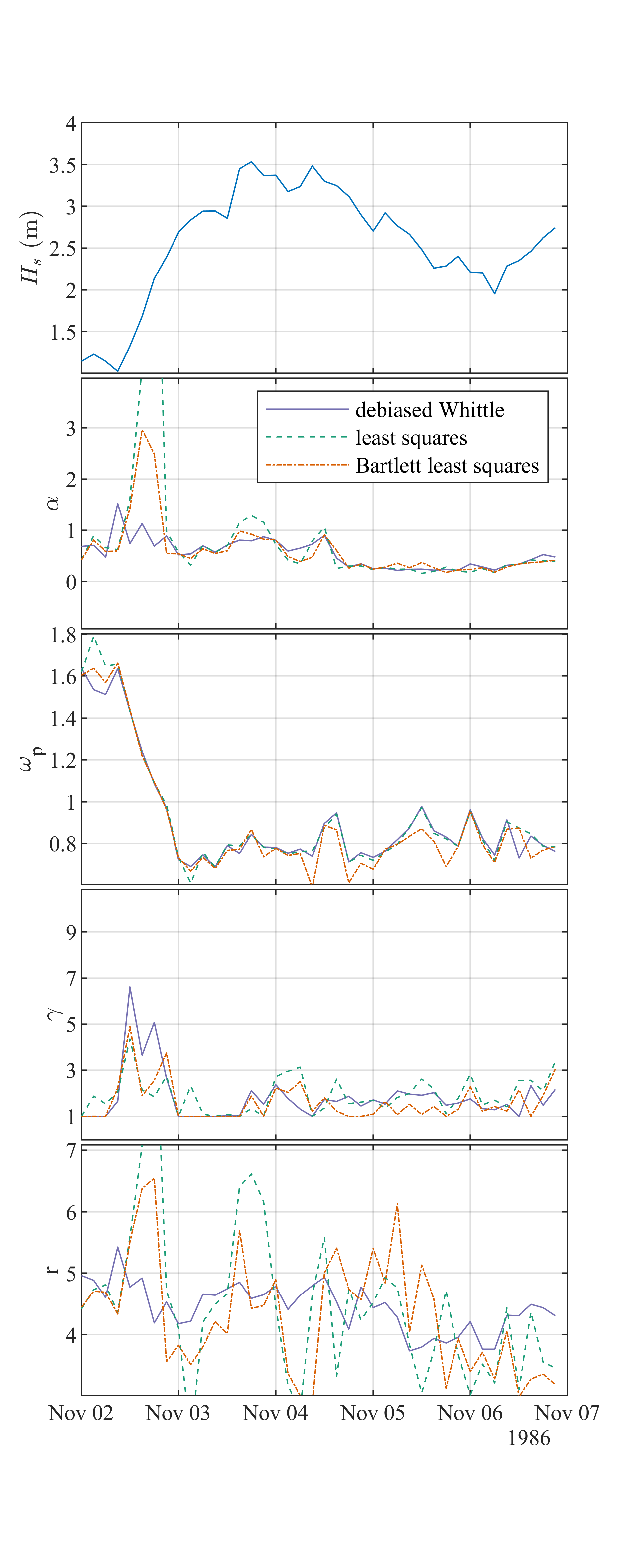}
        \caption{Parameter estimates using least squares, Bartlett least squares and de-biased Whittle likelihood.}
        \label{fig:maui:ls}
    \end{subfigure}
    \hfill
    \begin{subfigure}[t]{0.475\textwidth}
        \centering
        \includegraphics{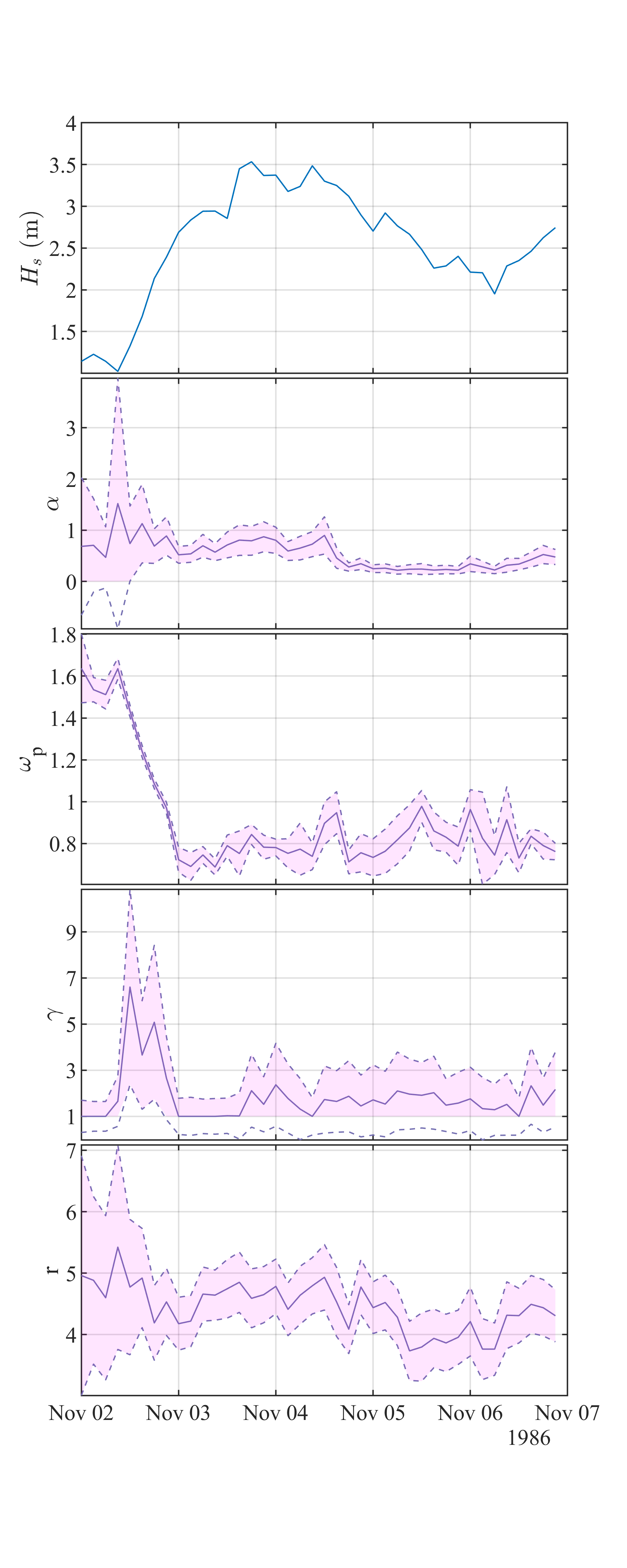}
        \caption{Parameter estimates using de-biased Whittle likelihood with approximate 95\% confidence intervals.}
        \label{fig:maui:dw}
    \end{subfigure}%
    \caption{Plots of $H_s$ and the estimated parameters over time, from the second to the seventh of November, 1986. Estimates are from 20 minute records starting every 3 hours and are plotted at the start time of each record. Note that some of the plots in (a) have been truncated. The maximum of the least squares estimate for $\alpha$ was 9.37 and maximum and minimum for $r$ were 9.66 and 2.28.}
    \label{fig:maui:Boxplot}
\end{figure}
\begin{figure}[htp!]
    \centering
    \begin{subfigure}[t]{\textwidth}
        \centering
        \includegraphics{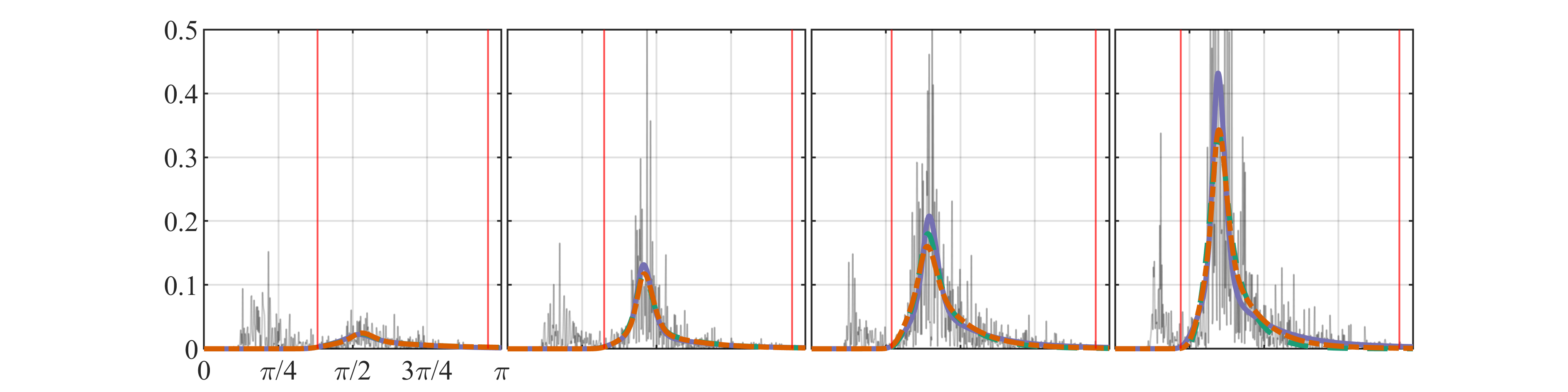}
        \caption{Periodogram and the expected periodogram of the fitted models.}
        \label{fig:maui:diag:I}
    \end{subfigure}
    \hfill
    \begin{subfigure}[t]{\textwidth}
        \centering
        \includegraphics{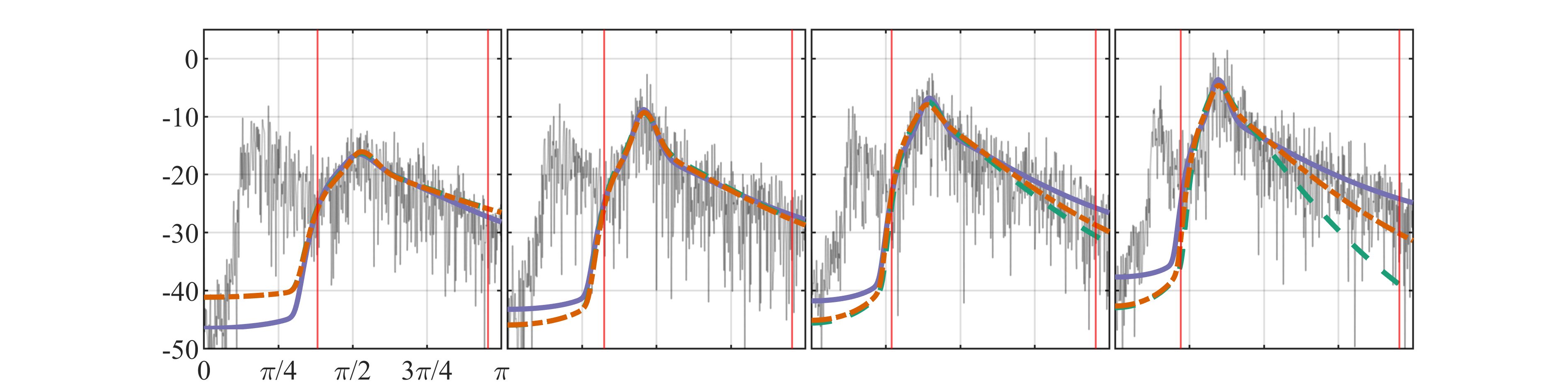}
        \caption{Periodogram and the expected periodogram of the fitted models on a decibel scale.}
        \label{fig:maui:diag:IdB}
    \end{subfigure}
    \hfill
    \begin{subfigure}[t]{\textwidth}
        \centering
        \includegraphics{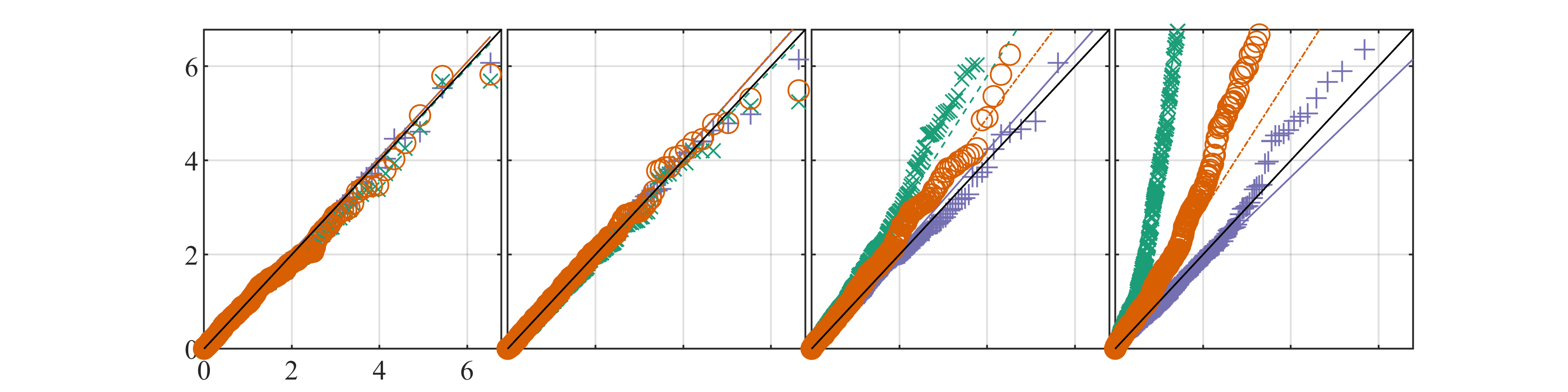}
        \caption{Q-Q plots of $I(\omega)/\EE{I(\omega)\mid\theta}$ for parameter vectors $\theta$ obtained from each fitting procedure.}
        \label{fig:maui:diag:QQ}
    \end{subfigure}
    \caption{Plot of periodograms and fitted models on normal and decibel scale, alongside diagnostic Q-Q plots. Each column is from records recorded on the fourth of November, 1986, starting at 09:00, 12:00, 15:00 and 18:00 respectively. Colours and line styles are used to denote different fitting techniques in the same manor as Figure \ref{fig:maui:ls}. In (c), reference lines are added to each of the fitting methods to aid comparison to the black line denoting $x=y$. The symbols ``$+$'', ``$\times$'' and ``$\circ $'' denote quantiles from the de-biased Whittle likelihood, least squares estimates, and Bartlett least squares estimates, respectively. Note that the plots in (c) have been truncated in the y axis so that they have a 1:1 aspect ratio, so some quantiles are not shown. The untruncated figures can be found on GitHub \citep{Grainger2021a}.}
    \label{fig:maui:QQ}
\end{figure}
\par
From Figure \ref{fig:maui:ls}, we can see that, especially for the tail index, the de-biased Whittle estimates seem more stable than both the least squares and Bartlett least squares estimates.
Figure \ref{fig:maui:dw} shows the de-biased Whittle likelihood estimates and approximate 95\% point-wise confidence intervals, calculated using the technique described in Section \ref{sec:confidence}.
It is worth noting that these confidence intervals are based on the asymptotic distribution of the estimator, in particular, we assume asymptotic normality of the estimator.
This is especially problematic for values near to or on the boundary of the parameter space, as the confidence interval may include values outside the boundary.
To represent this, in Figure \ref{fig:maui:dw}, we only shade regions of the confidence interval that are within the parameter space.
Such confidence intervals enable us to better understand and communicate the uncertainty surrounding parameter estimates. 
This uncertainty, we argue, should be considered when using these parameters as inputs for other related models.
\par
From Figure \ref{fig:maui:dw} we can see that the parameters behave broadly as expected when moving through the storm. We see a steady evolution in peak frequency throughout the growth of the wind sea component, alongside large values of $\gamma$ as the storm is developing. After this point, the values of $\gamma$ reside close to 1.
We can also see that there is a greater level of uncertainty when $H_s$ was low. This is because there is very little energy in the wind-sea component. However, as the storm evolves, this uncertainty quickly reduces.
The uncertainty surrounding the peak enhancement factor remains fairly high through the second half of the first day, even when the uncertainty around other parameters has reduced. This is because when the true value of $\gamma$ is large, at frequencies around the peak the periodogram has more variance, and it is these frequencies that contain the information about $\gamma$.
\par
A useful diagnostic tool for spectral models can be constructed by noticing that (under certain technical conditions) $I(\omega)/\EE{I(\omega)\mid\theta}$ should be approximately exponentially distributed with mean 1 \citep{Brockwell2006}, where $I(\omega)$ denotes the \textit{periodogram} (this result is specific to the periodogram and is not true in general for non-parametric spectral density estimators).
Therefore, if we calculate this ratio for a given model at each Fourier frequency, we can then compare the quantiles of these fractions to the quantiles of an exponential distribution with mean 1 using a Q-Q plot.
Figure \ref{fig:maui:diag:I} and \ref{fig:maui:diag:IdB} show the periodograms and fitted models on both the standard scale and the decibel scale respectively, for 4 time periods during the early development of the storm.
Figure \ref{fig:maui:diag:QQ} shows the corresponding Q-Q plots, which demonstrate how well each of the models fits the data.
The closer the points on the Q-Q plot are to the reference line $y=x$, the better.
We can see from Figure \ref{fig:maui:diag:QQ} that the de-biased Whittle likelihood consistently outperforms the other techniques.
Plots for the remaining sea states are available on GitHub \citep{Grainger2021a}.

\section{Discussion and Conclusion}\label{sec:discussion}
    The de-biased Whittle likelihood has been shown to yield major improvements in both the bias and variance of estimated parameters for wind-generated waves.
    In particular, the tail decay index can be estimated to much greater levels of accuracy and precision than when using least squares techniques.
    Such an improvement will enable reliable tracking of the tail decay index's behaviour throughout the course of meteorological events, allowing oceanographers to gain fresh insights into the behaviour of wind-generated waves.
    We have also demonstrated some improvement in the estimation of the peak enhancement factor.
    The de-biased Whittle estimator recovers estimates that are of similar quality to full maximum likelihood, which can be thought of as optimal.
    Since information about the peak enhancement factor is contained in a small region around the peak frequency, it is not surprising that it is so hard to estimate.
    Because there is significant variability in estimates of the peak enhancement factor, it is essential that we can describe uncertainty surrounding parameter estimates when performing an analysis.
    For this reason, the development of computationally efficient techniques for quantifying uncertainty surrounding the estimated parameters is important.
    We have shown that the method presented by \cite{Sykulski2016a} for estimating the variance of the de-biased Whittle estimator can be modified so that it can be computed using 2D Fast Fourier Transforms.
    Combining this with an analytical approach for calculating derivatives, we are able to calculate the uncertainty in parameter estimates accurately and quickly.
    In addition, differencing can be used so that we can cope with high sampling frequency data, which tend to be correlated in the frequency domain.
    \par
    As we discussed previously, when performing simulations we chose to use circulant embedding \citep{Davies1987} to obtain realisations of Gaussian processes with the desired covariance matrix. 
    This is different to the standard method used for simulating linear ocean waves, presented by \cite{Tucker1984} (or to the adapted version due to \cite{Merigaud2018}, which is preferable).
    The first difference is that the standard waves method is only approximate, and does not exactly simulate a Gaussian process with the desired covariance (equivalently spectral density).
    The second is that these methods do not account for the aliasing that we would expect to be present in a record, essentially treating the spectral density as if it is zero beyond the Nyquist frequency.
    Though for many applications this will not matter, the methods discussed in this paper will be sensitive to this difference. If the simulated record does not have the aliasing that should be present, then the parameters estimated using the de-biased Whittle likelihood will seem biased (and standard Whittle will often seem better).
    This is important because if a method such as \cite{Tucker1984} was used to perform the simulation study described in this paper, the results would be different, as likelihood based estimation will be sensitive to this problem (which could be thought of as model miss-specification, as we are trying to fit a model with a non-zero density beyond the Nyquist to a process that has been simulated with no density above the Nyquist).
    Least squares is somewhat invariant to this problem, as it mainly effects frequencies where the spectral density function is small, and least squares is not heavily influenced by such frequencies.
    \par
    We have also developed a \matlab toolbox, implementing the methods discussed in this paper.
    This is available, alongside code to generate the figures in this paper and additional supplementary figures, on GitHub \citep{Grainger2021a}.
    The toolbox contains code to perform each of the fitting techniques discussed in this paper on arbitrary processes, as well as a function implementing the generalised JONSWAP (including first and second derivatives), which can then be used straight out the box.
    On top of this, the user may provide any spectral density function they wish, provided it satisfies the assumptions in Appendix~\ref{append:assumption}, and then use the toolbox to obtain parameter estimates from observations.
    Alongside this, an implementation of circulant embedding is provided, enabling exact simulation of a desired Gaussian process.
    \par
    Though the de-biased Whittle likelihood has been seen to perform well in the estimation of wind-seas, we have yet to fully explore its potential when we wish to describe multi-modal seas (e.g. including one or more swell components).
    In this paper, we were able to describe the wind-sea component of such sea states by first removing swell with a high pass filter.
    However, it would be preferable to develop and fit models that were bi-modal themselves, avoiding the need for partitioning schemes to determine which frequencies should be filtered.
    Indeed, such a procedure could be extended to describe seas with any number of components, using model selection to determine the number of components that are actually present.
    Such techniques could also allow for the development of model-based partitioning schemes, that could to separate overlapping wind-sea and swell components.
    In the development of such multi-modal models, it would be important to consider the interactions between different component weather systems. In particular, techniques such as higher-order spectral analysis could be used to determine if any non-linear interactions were present and help characterise them.
    Such interactions could then also be parameterised and fitted using similar techniques to those discussed in this paper.
    \par
    Another important aspect of wind-generated waves is their directional characteristics. 
    One approach to describing these is to assume some dispersion relation between wave-number and frequency, and then look at frequency direction spectra, which can be estimated, for example, from heave-pitch-roll buoys \citep{Longuet-Higgins1963}.
    However, it may preferable to model the three recorded series instead, requiring the use of a multivariate extension to the de-biased Whittle likelihood. These are all important avenues of further investigation.
    \par
    In summary, we have demonstrated that, by using the de-biased Whittle likelihood, we are able to obtain more accurate and precise estimates of parameters for the wind sea component of a wind-generated wave process in $O(N\log N )$ time.
    Using differencing, we are able to overcome correlations in the periodogram that are common in data sampled at a high frequency.
    Furthermore, we have described a procedure which can be used to obtain the variance of such estimates in $O(N^2\log N )$ time.
    Together these developments will improve the tools available to practitioners, both in terms of fitting models to data and describing the associated uncertainty.
\section*{Acknowledgements}
J. P. Grainger's research was funded by the UK Engineering and Physical Sciences Research Council (Grant EP/L015692/1) and JBA Trust. 
The work of A. M. Sykulski was funded by the UK Engineering and Physical Sciences Research Council (Grant EP/R01860X/1).
We thank Rob Lamb of JBA and Lancaster University, UK for discussions.
We also thank two anonymous reviewers for their feedback.

\appendix
\section{Derivatives of ocean wave spectra}\label{append:deriv}
When calculating the derivatives of the generalised JONSWAP form, special care must be taken around the peak frequency and zero frequency, to check that the generalised JONSWAP form is actually differentiable. We do not show this here, but for the first derivatives there is no issue (though this is not the case for the second derivatives, which do not all exist at the peak).
Consider the generalised JONSWAP spectral form, $f_G(\omega\mid\theta)$. We may write, for $\omega>0$,
\begin{align}
    \frac{\partial}{\partial\omega}\fbit&=\frac{\alpha}{2}\omega^{-r}\expbit\gamma^{\dbit}\left[\frac{r}{\omega_p}\left(\frac{\omega}{\omega_p}\right)^{-s-1}-\frac{r}{\omega}-\frac{\dbit\cdot(\omega-\omega_p)\log\gamma}{\omega_p^2\sigbit^2}\right]\\
    &=\fbit\left[\frac{r}{\omega_p}\left(\frac{\omega}{\omega_p}\right)^{-s-1}-\frac{r}{\omega}-\frac{\dbit\cdot(\omega-\omega_p)\log\gamma}{\omega_p^2\sigbit^2}\right]\\
    &\rightarrow 0 \text{ as }\omega\rightarrow 0^+,\\
    \frac{\partial}{\partial\alpha}\fbit&=\omega^{-r}\exp\left \{-\frac{r}{s}\left(\frac{\omega}{\omega_p}\right)^{-s}\right \}\gamma^{\dbit}/2\\
    &=\fbit/\alpha\\
    &\rightarrow 0 \text{ as }\omega\rightarrow 0^+,\\
    \pard{\omega_p}\fbit&=\frac{\alpha}{2}\omega^{-r}\expbit\gamma^{\dbit}\left[\dbit\log(\gamma)\frac{\omega}{\sigbit^2}\left(\frac{\omega-\omega_p}{\omega_p^3}\right)-\frac{\omega r}{\omega_p^2}\left(\frac{\omega}{\omega_p}\right)^{-s-1}\right]\\
    &=\fbit\left[\dbit\log(\gamma)\frac{\omega}{\sigbit^2}\left(\frac{\omega-\omega_p}{\omega_p^3}\right)-\frac{\omega r}{\omega_p^2}\left(\frac{\omega}{\omega_p}\right)^{-s-1}\right]\\
    &\rightarrow 0 \text{ as }\omega\rightarrow 0^+,\\
    \frac{\partial}{\partial\gamma}\fbit&=\alpha\omega^{-r}\expbit\dbit\gamma^{\dbit-1}/2\\
    &= \fbit\dbit/\gamma\\
    &\rightarrow 0 \text{ as }\omega\rightarrow 0^+,\\
    \pard{r}\fbit&=\alpha\omega^{-r}\expbit\gamma^{\dbit}\left[-\log(\omega)-\frac{1}{s}\left(\frac{\omega}{\omega_p}\right)^{-s}\right]/2\\
    &=\fbit\left[-\log(\omega)-\frac{1}{s}\left(\frac{\omega}{\omega_p}\right)^{-s}\right]\\
    &\rightarrow 0 \text{ as }\omega\rightarrow 0^+.
\end{align}
By the chain rule, for $\omega<0$
\begin{align}
    \pard{\omega}\fbit&=-\pard{\omega}f_G(\minusparbit),\\
    \negderiv{\alpha},\\
    \negderiv{\omega_p},\\
    \negderiv{\gamma},\\
    \negderiv{r}.
\end{align}
Finally, for $\omega=0$, since $f_G(0,\alpha,\omega_p,\gamma,r)=0$, we may write
\begin{align}
    \pard{\omega}f_G(\omega,\mid\theta)\mid_{\omega=0}&=0,\\
    \pard{\alpha}f_G(0,\mid\theta)&=0,\\
    \pard{\omega_p}f_G(0,\mid\theta)&=0,\\
    \pard{\gamma}f_G(0,\mid\theta)&=0,\\
    \pard{r}f_G(0,\mid\theta)&=0.
\end{align}
Therefore, we see that $f_G$ has continuous first derivatives.
\begin{proposition}\label{prop:unifconv}
    For all $x\in\{\omega,\alpha,\omega_p,\gamma,r\}$ the series
    \begin{align}
        \sum_{k=-\infty}^\infty \frac{\partial}{\partial x}f_G\left(\omega+\frac{2\pi k}{\Delta}\mid\theta\right)
    \end{align}
converges uniformly.
\end{proposition}
\begin{proof}
    Firstly, due to symmetry, we may consider the series
    \begin{align}
        \sum_{k=L}^\infty \frac{\partial}{\partial x}f_G\left(\omega+\frac{2\pi k}{\Delta}\mid\theta\right)
    \end{align}
    where $L$ is such that $\omega_p<\frac{\pi}{\Delta}+\frac{2\pi L}{\Delta}$.
    Now write
    \begin{align}
        g_{k}(\omega\mid\theta)&=f_G\left(\omega+\frac{2\pi k}{\Delta}\mid\theta\right)\quad\text{and}\\
        g_{k,x}(\omega\mid\theta)&=\frac{\partial}{\partial x}f_G\left(\omega+\frac{2\pi k}{\Delta}\mid\theta\right).
    \end{align}
    Next notice that
    \begin{align}
        |g_k(\omega\mid\theta)|&\leq\alpha\gamma\left(\omega+\frac{2\pi k}{\Delta}\right)^{-r}\\
        &\leq\alpha\gamma\left(\frac{-\pi}{\Delta}+\frac{2\pi k}{\Delta}\right)^{-r}\\
        &=\alpha\gamma\left(\frac{\pi}{\Delta}\right)^{-r}\left(2k-1\right)^{-r}.
    \end{align}
    Let $M_k=\alpha\gamma\left({\pi}/{\Delta}\right)^{-r}\left(2k-1\right)^{-r}$. It now suffices to notice that each of the partial derivatives can be written as $f_G(\omega\mid\theta)$ multiplied by some other function. Furthermore, each of these functions can be bounded. Therefore we may write
    \begin{align}
        |g_{k,x}(\omega\mid\theta)|\leq C_x M_k,
    \end{align}
    for some constants $C_x>0$. Finally
    \begin{align}
        \sum_{k=L}^\infty M_k <\infty,
    \end{align}
    so by the Weierstrass M-test, we have uniform convergence. By extension we can see that
    \begin{align}
        \sum_{k=-\infty}^\infty \frac{\partial}{\partial x}f_G\left(\omega+\frac{2\pi k}{\Delta}\mid\theta\right)
    \end{align}
    is also uniformly convergent.
\end{proof}
\begin{proposition}
    For all $x\in\{\omega,\alpha,\omega_p,\gamma,r\}$ 
    \begin{align}
        \frac{\partial}{\partial x}f_\Delta\left(\omega\mid\theta\right)&=\sum_{k=-\infty}^\infty \frac{\partial}{\partial x}f_G\left(\omega+\frac{2\pi k}{\Delta}\mid\theta\right).
    \end{align}
    Furthermore, $ \frac{\partial}{\partial x}f_\Delta\left(\omega\mid\theta\right)$ is continuous.
\end{proposition}
\begin{proof}
    This first part follows from Proposition \ref{prop:unifconv}, the convergence of $\sum_{k=-\infty}^\infty f_G\left(\omega+2\pi k/\Delta\mid\theta\right)$, the continuous differentiability of $f_G$, and Theorem~4.4.20 of \cite{Trench2013}. The continuity follows from the continuity of the derivatives of $f_G$ and the uniform limit theorem.
\end{proof}
\section{Computing the variance of the first derivative}\label{append:var}
\cite{Sykulski2016a} decompose the variance of the first derivative of the de-biased Whittle likelihood as follows:
\begin{align}
    \var{\pard{\theta_i}\ell_{DW}(\theta)}&=\sum_{j=1}^N\sum_{k=1}^N a_{ij}(\theta) a_{ik}(\theta) \cov{I(\omega_j),I(\omega_k)},
\end{align}
where $\omega_j, \omega_k$ denote Fourier frequencies and
\begin{align}
    a_{ij}(\theta)&=\frac{\partial \bar f_N(\omega_j\mid\theta)}{\partial \theta_i}\frac{1}{\bar f^2_N(\omega_j\mid\theta)}.
\end{align}
To estimate the variance, they propose using
\begin{align}
    \varhat{\pard{\theta_i}\ell_{DW}(\theta)}&=\sum_{j=1}^N\sum_{k=1}^N a_{ij}(\hat\theta_{DW}) a_{ik}(\hat\theta_{DW}) \covhat{I(\omega_j),I(\omega_k)}.
\end{align}
For ocean wave models, $a_{ij}(\hat\theta_{DW})$ can be easily computed by using the results from Section  \ref{append:deriv}. So, from \cite{Sykulski2016a}, we are interested in computing
\begin{align}
    \hat{\text{cov}}\left( I(\omega_j),I(\omega_k) \right)&= \left| \frac{\Delta}{2\pi N}\int_{-\pi/\Delta}^{\pi/\Delta} f_\Delta(\omega'\mid\theta)D_N(\Delta(\omega_j-\omega))D_N^*(\Delta(\omega_k-\omega))\de \omega'\right|^2,
\end{align}
where
\begin{align}
    D_N(v)&=\frac{\sin{(Nv/2)}}{\sin(v/2)}e^{-iv(N+1)/2}.
\end{align}
To do this efficiently, first note that
\begin{align}
    D_N(v)&=\frac{\sin{(Nv/2)}}{\sin(v/2)}e^{-iv(N+1)/2}
    =\frac{e^{iNv/2}-e^{-iNv/2}}{e^{iv/2}-e^{-iv/2}}e^{-iv(N+1)/2}\\
    &=\frac{e^{-iNv/2}}{e^{-iv/2}}\frac{e^{iNv}-1}{e^{iv}-1}e^{-iv(N+1)/2}
    =e^{-i(N-1)v/2}e^{-iv(N+1)/2}\sum_{s=0}^{N-1}e^{isv} & \text{ by geometric series}\\
    &=e^{-iNv}\sum_{s=0}^{N-1}e^{isv}.
\end{align}
It is also convenient to write
\begin{align}
    D^*_N(v) &= \frac{\sin{(Nv/2)}}{\sin(v/2)}e^{iv(N+1)/2}
    =e^{iv}\sum_{s=0}^{N-1}e^{isv}.
\end{align}
Now, consider the function
\begin{align}
    h_{jk}(\omega)&=f_\Delta(\omega\mid\theta) D_N(\Delta(\omega_j-\omega)) D_N^*(\Delta(\omega_k-\omega))\\
    &=f_\Delta(\omega\mid\theta)\left(e^{-iN(\Delta(\omega_j-\omega))} \sum_{s=0}^{N-1}e^{is(\Delta(\omega_j-\omega))}\right) \left(e^{i(\Delta(\omega_k-\omega))} \sum_{r=0}^{N-1}e^{ir(\Delta(\omega_k-\omega))}\right)\\
    &=e^{-i\Delta(N\omega_j-\omega_k)} \sum_{s=0}^{N-1}\sum_{r=0}^{N-1} e^{is\Delta\omega_j}e^{ir\Delta\omega_k} \left[f_\Delta(\omega\mid\theta )e^{i\Delta(N-1)\omega} e^{-i\Delta(s+r)\omega}\right],
\end{align}
where we have rearranged for later convenience. We can now see, by linearity of integration, that
\begin{align}
    \int_{-\pi/\Delta}^{\pi/\Delta}h_{jk}(\omega)\de\omega &= e^{-i\Delta(N\omega_j-\omega_k)} \sum_{s=0}^{N-1}\sum_{r=0}^{N-1} e^{is\Delta\omega_j}e^{ir\Delta\omega_k} \int_{-\pi/\Delta}^{\pi/\Delta}f_\Delta(\omega\mid\theta)e^{i\Delta(N-1)\omega}e^{-i\Delta(s+r)\omega}\de\omega.
\end{align}
Thus for $r=0,\ldots,N-1$ and $s=0\ldots,N-1$ we must calculate
\begin{align}
    \int_{-\pi/\Delta}^{\pi/\Delta} f_\Delta(\omega\mid\theta) e^{i\Delta(N-1)\omega} e^{-i\Delta(r+s)\omega}\de\omega.
\end{align}
Notice that by letting $t=r+s$ we need to calculate the following integral for $t=0,\ldots,(2N-1)-1$
\begin{align}
    \int_{-\pi/\Delta}^{\pi/\Delta} f_\Delta(\omega\mid\theta) e^{i\Delta(N-1)\omega} e^{-i\Delta t\omega}\de\omega.
\end{align}
For clarity let
\begin{align}
    q(\omega)&=f_\Delta(\omega\mid\theta) e^{i\Delta(N-1)\omega}
\end{align}
then we need to compute
\begin{align}
    Q(t)&=\int_{-\pi/\Delta}^{\pi/\Delta} q(\omega) e^{-i\Delta t\omega}\de\omega,\label{eq:Q(t)def}
\end{align}
for $t=0,\ldots,2N-2$.
We notice that this is a Fourier transform and we can obtain an approximation of this integral at each of the required t by doing an fft on the relevant length $2N-1$ sequence (the details of this are discussed in Section  \ref{append:var:approx}). This means that we only need to do one Fourier transform, and can then sub this into the previous sums.
Now we require
\begin{align}
    \int_{-\pi/\Delta}^{\pi/\Delta}h_{jk}(\omega)\de\omega &= e^{-i\Delta(N\omega_j-\omega_k)} \sum_{s=0}^{N-1}\sum_{r=0}^{N-1} e^{is\Delta\omega_j} e^{ir\Delta\omega_k} Q(r+s)\\
    &=e^{-i\Delta(N\omega_j-\omega_k)} \sum_{s=0}^{N-1}\sum_{r=0}^{N-1} \tilde Q(r,s)e^{i\Delta 2\pi (rk+sj)/N\Delta}\\
    &=e^{-i\Delta(N\omega_j-\omega_k)} \sum_{s=0}^{N-1}\sum_{r=0}^{N-1} \tilde Q(r,s)e^{i2\pi (rk+sj)/N}\\
\end{align}
where $\tilde Q(r,s)=Q(r+s)$. This is a 2D Fourier transform, and can be computed efficiently for $j,k=0,\ldots,N-1$. This means we take $O(N^2\log N)$ time. Importantly, libraries exist to compute this very quickly in most programming languages.
Finally
\begin{align}
    \hat{\text{cov}}\left( I(\omega_j),I(\omega_k) \right)&= \left|\frac{\Delta}{2\pi N}\int_{-\pi/\Delta}^{\pi/\Delta}h_{jk}(\omega)\de\omega\right|^2\\
    &=\left|\frac{\Delta}{2\pi N} e^{-i\Delta(N\omega_j-\omega_k)} \sum_{s=0}^{N-1}\sum_{r=0}^{N-1}\tilde Q(r,s)e^{i2\pi (rk+sj)/N}\right|^2\\
    &=\left|\frac{\Delta}{2\pi N}\sum_{s=0}^{N-1}\sum_{r=0}^{N-1}\tilde Q(r,s)e^{i2\pi (rk+sj)/N}\right|^2.
\end{align}
\subsection{Approximating $Q(t)$}\label{append:var:approx}
Note that we must still approximate $Q(t)$ at $t=0,\ldots,2N-2$. 
We aim to approximate the integral in \eqref{eq:Q(t)def}. To achieve this, we use the Riemann sum given by
\begin{align}
    \bar Q(t) &= \frac{2\pi}{M\Delta}\sum_{j=0}^{M-1} q(2\pi j/M\Delta)e^{-i(t\Delta)(2\pi j/M\Delta)}.
\end{align}
An FFT will produce values of $\bar Q(t)$ for $t=0,1,\ldots,(M-1)$. Therefore, provided $M\geq 2N-1$ we can obtain approximations of the desired integrals.
\section{Assumptions for de-biased Whittle}\label{append:assumption}
    Theorem 1 of \cite{Sykulski2016a} gives assumptions for the optimal convergence of the de-biased Whittle likelihood estimator. The first of these assumptions is that the parameter space is compact with non-null interior and that the true value of the parameter lies in the interior of the parameter space. This is not strictly satisfied by the generalised JONSWAP spectral density; however, for physical reasons the parameter space can be restricted so that it is compact. The greater problem is that a value of $\gamma=1$, which corresponds to the fully developed sea described by \cite{Pierson1964}, may occur in nature and would violate this assumption. In Section  \ref{sec:comparison:robust} we demonstrate that, in this case, the de-biased Whittle likelihood estimator still performs well.
    \par
    The second assumption is that the spectral density of the aliased process is bounded above and is bounded below by some positive real number. This is satisfied by the generalised JONSWAP spectral density. The aliased spectrum is bounded below by the non-aliased spectrum, and is bounded above by the variance of $X_t$ (which is finite).
    Therefore, the only frequency remaining to consider is zero, as the spectral density is zero when $\omega=0$. However, contributions from above the Nyquist frequency are positive, and as such the aliased spectral density at zero will also be positive.
    \par
    Assumption three relates to parameter identifiability. Informally, this requires the aliased spectral density function to be different for different choices of $\theta$.
    Intuitively, provided the sampling interval is sufficiently small (so that the peak frequency is smaller than the Nyquist), then each of the parameters is changing the shape of the generalised JONSWAP spectrum in a different way, such that the parameters will in general be identifiable for a sufficient sample size.
    \par
    The fourth assumption states that the aliased spectral density function must be continuous in $\theta$ and Riemann integrable in $\omega$. This is satisfied for the generalised JONSWAP spectral form as it is continuous in $\theta$ and $\omega$.
    \par
    Assumption five states that the expected periodogram has two continuous derivatives in $\theta$, and that these derivatives are bounded uniformly for all $N$. Furthermore the first derivatives are required to have $\Theta(N)$ non-zero frequencies.
    Strictly speaking, the generalised JONSWAP is not twice differentiable (the second derivative does not exist at the peak, due to the step function $\sigbit$).
    However, a simple adaptation can be made, by replacing $\sigbit$ with
    \begin{align}
        \bar \sigma (\omega\mid\theta) &= \sigma_1+(\sigma_2-\sigma_1)\left(\frac{1}{2}+\frac{1}{\pi}\arctan\left(C\left(\omega-\omega_p\right)\right)\right),
    \end{align}
    where $C\in(0,\infty)$ is chosen to be large. This has the advantage of having continuous second derivatives, and is essentially equivalent to the generalised JONSWAP, because $C$ can be chosen such that it would be impossible to distinguish between the two models from observed data. 
    Since this part of the generalised JONSWAP was developed empirically, there is no practical difference in using the step function over this reformulation. 
    Indeed, in keeping with the general philosophy of statistical modelling --- ``all models are wrong, but some models are useful'' --- we suggest that this alternative formulation is just as appropriate as the generalised JONSWAP, but more useful here as it allows us to show this assumption is satisfied.
    By similar arguments to those presented in Appendix~\ref{append:deriv} and Appendix~\ref{append:var}, we can see that the aliased spectral density has continuous second derivatives. The autocovariance then has continuous second derivatives by a similar argument to that in Section~\ref{sec:confidence}, and noting that the second derivative of the aliased spectral density function is integrable, and so therefore the second derivative of the autocovariance must be continuous (as they can be shown to be Fourier pairs). Therefore, the expected periodogram also has continuous second derivatives by linearity of derivatives and the fact that linear combinations of continuous functions are continuous.
    These derivatives are also bounded uniformly for all $N$ due to the compactness of the set of frequencies $\Omega$ and the parameter space $\Theta$.
    \par
    The final assumption states that the process in question is fourth-order stationary with finite fourth order moments and absolutely summable fourth order cumulants.
    Clearly this is true for a Gaussian process (as second order stationarity implies strict stationarity for Gaussian) and it is also true for some non-linear processes, such as the class of non-linear processes discussed by \cite{Sykulski2016a}.
    \par
    Finally, in our simulations the estimator based on the de-biased Whittle likelihood performs in broad agreement with the theory, for example we observed desirable properties such as root $N$ convergence when exploring different values of $N$.
    We also did not find any problems with local minima during optimisation for any of the record lengths considered, suggesting that the parameters of the generalised JONSWAP form are indeed identifiable in practice for sufficiently long records.
\newpage
\bibliography{bib/Applied,bib/code, bib/Bloomfield}

\end{document}